%% file: main.tex
\documentclass{ws-ijmpa}

\usepackage[super]{cite}
\usepackage{xcolor}
\usepackage[verbose,hypertexnames=false]{hyperref}
\usepackage{subcaption}
\usepackage{xfrac}
\hypersetup{colorlinks=false,allbordercolors=blue,pdfborderstyle={/S/U/W 1}}
\include{mymacros}

\begin{document}


%
\catchline{}{}{}{}{}
%

\title{Duality of Navier-Stokes to a one-dimensional system}

\author{Alexander Migdal\footnote{
amigdal@ias.edu}
}

\address{Member, School of Mathematics,\\ Institute for Advanced Study\footnote{
1 Einstein Drive, Princeton, New Jersey, 08540, USA}}

\maketitle

\begin{history}
\received{Day Month Year}
\revised{Day Month Year}
\accepted{Day Month Year}
\published{Day Month Year}
\end{history}

\begin{abstract}
The Navier–Stokes (NS) equations describe fluid dynamics through a high-dimensional, nonlinear partial differential equations (PDEs) system. Despite their fundamental importance, their behavior in turbulent regimes remains incompletely understood, and their global regularity is still an open problem. Here, we reformulate the NS equations as a nonlinear equation for the momentum loop $\vec P(\theta, t)$, effectively reducing the original three-dimensional PDE to a one-dimensional problem. A key result of this reformulation is the derivation of a \textit{No Explosion Theorem}, establishing that finite-time singularities do not occur under stochastic initial conditions. We also present an explicit analytical solution—the Euler ensemble—which describes the universal asymptotic state of decaying turbulence and is supported by numerical simulations and experimental validation.
\end{abstract}


\ccode{PACS numbers: 03.65.$-$w, 04.62.+v}
\section{Summary for QFT/String theorists}

The turbulence problem used to be a clean, unsolved problem of theoretical physics \cite{Feynman}: find a mathematical description of the chaotic circular motion of a fluid subject to the \NS{} equations. The solution should explain the origin of the statistical distribution in the nonlinear PDE without any random forcing and provide a method to compute the statistical properties, such as observed energy spectra and exponents of the time decay of turbulent kinetic energy.
\begin{eqnarray}
    &&\partial_t \vec v = -\nu \nabla \times \vec \omega - \vec v\times \vec \omega - \vec \nabla \left( p + \frac{\vec v^2}{2} \right);\\
    && \nabla \cdot \vec v =0;\\
    &&\vec \omega = \vec \nabla \times \vec v
\end{eqnarray}
Stated like this, the only way to solve this problem is to derive and solve the equation for the evolution of the statistical distribution of the velocity field, and find an asymptotic trajectory of this distribution, converging to a zero velocity (result of total dissipation of kinetic energy).
The attempt to find such an asymptotic solution was made by Eberhard Hopf \cite{Hopf19}, who derived the evolution equation for the generating functional
\begin{eqnarray}
    Z[J, t] = \VEV{\exp{\int d^3 r \vec J(\vec r) \cdot \vec v(\vec r, t)}}_{NS}; \vec \nabla \cdot \vec J(\vec r)=0;
\end{eqnarray}
Here, averaging goes over all solutions of the \NS{} equations with certain noisy initial conditions. In a real world, the would be the Gibbs distribution
\begin{eqnarray}
    && W[v]\propto \exp{- \int d^3 r \frac{\vec v^2(\vec r)}{2 T_0}};\\
    && Z[J, 0] = \exp{+ T_0 \int d^3 r \frac{\vec J^2(\vec r)}{2}}
\end{eqnarray}
This functional would decrease at a large source only for an extra imaginary factor in $\vec J$. These questions weren't addressed in the Hopf approach as it was used just as a generating functional.
The Hopf equation has the form
\begin{eqnarray}
    \pd{t} Z[J, t] = \hat H\left[J, \ff{J}\right] Z[J, t]
\end{eqnarray}
This equation is too general to solve. It is not using important geometric properties of the velocity field: it is not just the field variable but also the velocity of every particle in the flow.

The velocity field and the vector potential in the electromagnetic field in the Landau gauge look similar. This analogy is purely superficial, as there is no gauge invariance in fluid dynamics. There is, however, another, very powerful invariance, which leads to an exact solution for the statistical distribution described in this paper.

This invariance manifests itself in the loop space formulation of the fluid mechanics\cite{M93, M23PR}. We introduce a loop functional (an abelian Wilson loop)
\begin{eqnarray}
    \Psi[C,t] = \VEV{\exp{\frac{\I \oint d \vec r \cdot \vec v(\vec r, t)}{\nu}}}_{NS}
\end{eqnarray}
The area derivatives of this functional bring down vorticity $\vec \omega = \vec \nabla \times \vec v$, which is an analog of the field strength in QED.
\begin{eqnarray}
  \nu  \fbyf{\Psi[C,t]}{\vec \sigma} = \I \VEV{\vec \omega\exp{\frac{\I \oint d \vec r \cdot \vec v(\vec r, t)}{\nu}} }_{NS}
\end{eqnarray}
The loop equation \cite{M93, M23PR} replaces the Hopf equation
\begin{eqnarray}
    &&  -\I\nu\partial_t \Psi[C,t] =\oint d \vec C(\theta) \cdot \hat L(\theta)\Psi[C,t]
\end{eqnarray}
where $\hat L[\theta]$ is an operator acting on the loop $C(.)$.

The powerful symmetry mentioned above is the \textbf{translation invariance in loop space}. Due to this invariance, the operator $\hat L(\theta)$ involves only functional derivatives $\ff{C(.)}$ but not the loop $C$ by itself. 
\begin{eqnarray}
    \hat L(\theta) = \hat L\left[\ff{C(.)}|\theta\right]; 
\end{eqnarray}

The loop equation becomes the Schrödinger equation in loop space, with the Hamiltonian depending only on the canonical momenta but not on canonical coordinates.
Therefore, this equation is solved by a superposition of plane waves:
\begin{eqnarray}
    &&\Psi[C,t] = \VEV{\exp{ \I \int d \theta \vec C'(\theta)\cdot \vec P(\theta, t)}}_{P};\\
    &&-\I \partial_t \vec P(\theta, t) = \hat L\left[-\I \vec P'(\theta,t);\theta\right]
\end{eqnarray}

Explicit form of this operator (see below) is the \textbf{third degree homogeneous} function of $\vec P(\theta-0), \vec P(\theta+0) $ . This property allows to find an \textbf{exact} decaying solution
\begin{eqnarray}
    \vec P(\theta, t) = \frac{\vec F(\theta)}{\sqrt{2 \nu (t + t_0)}}
\end{eqnarray}
where $\vec F(\theta)$ satisfies \textbf{universal} fixed point equation

This equation becomes an algebraic equation relating $\vec F(\theta+0)$ to $\vec F(\theta-0)$ .
\begin{eqnarray}
     && \left((\Delta \vec{F})^2 -1\right)\vec F   =\Delta \vec{F} \left(\vec{F} \cdot \Delta \vec{F} +\I  \left( \frac{(\vec{F} \cdot \Delta \vec{F})^2}{\Delta \vec{F}^2}- \vec{F}^2\right)\right);\\
     && \vec F \equiv \frac{\vec F(\theta+0) + \vec F(\theta-0)}{2};\\
     &&\Delta \vec{F} \equiv \vec F(\theta+0) - \vec F(\theta-0);
\end{eqnarray}
We are interested in periodic solutions, which implies a certain parameter quantization condition. 

This periodic solution is described in detail below; it has a geometric meaning of a random walk on regular star polygons with quantized vertex angle $\beta = \frac{2 \pi p}{q}$ .
\begin{itemize}
                \item
                Here \(\vec{F}(\theta)\) is a \textbf{\textcolor{purple}{universal fractal curve}}, constructed as the limit \( N \to \infty \) of a regular star polygon \( \{q/p\} \) with vertices:
                $\vec{F}\left(\frac{2 \pi k}{N}\right) = \hat \Omega\cdot\frac{\{\cos(\alpha_k), \sin(\alpha_k), \I \cos(\frac{\beta}{2})\}}{2 \sin(\frac{\beta}{2})},
                $
                where:
               $
                    \beta = \frac{2\pi p}{q}$, $
                   \alpha_k = \beta\sum_{l=0}^k\sigma_l$ , $k = 1,\dots N, \quad N \to \infty $
                \item The parameters $ \hat{\Omega}\in SO(3), \frac{p}{q} \in \mathbb Q, \sigma_k = \pm 1 $ are \textbf{\textcolor{teal}{random}}, making \( \vec{P}(t,\theta) \) a \textbf{\textcolor{purple}{fixed stochastic trajectory}} of MLE.
                \item This solution is equivalent to a \textbf{\textcolor{red}{random walk on the set of regular star polygons}}.
                \item \textbf{\textcolor{green}{Validation:}} This solution has been verified using \Mathematica{} notebooks \cite{DecayTurb23} and rigorously tested in collaboration with mathematicians \cite{DeLellisInprep}.
                \item \textbf{\textcolor{orange}{Significance:}} This framework quantitatively \textbf{\textcolor{purple}{links turbulence to number theory}}, providing a fresh perspective on fluid dynamics.
            \end{itemize}
We suggested calling this set the Euler ensemble, as its probability measure involves the Euler totients (the number of fractions below one with a given denominator). Besides, this ensemble describes the inviscid limit of the NS equation, which is an Euler equation (up to so-called dissipation anomaly. 
This anomaly displays itself here as a WKB limit instead of a classical limit. 

We treat the Euler ensemble as a degenerate fixed point of the NS dynamics, replacing the energy surface (the fixed point of the statistical distribution for the Hamiltonian dynamics).

Assuming uniform covering of the Euler ensemble, we reduce the decaying turbulence to a number theory problem.

This problem can also be viewed as a special case of the string theory, with a discrete target space, given by the set of regular star polygons with unit side and corresponding radius $ R = \frac{1}{2 \sin(\pi p/q)}$. 
\begin{itemize}
        \item The loop functional in the Euler ensemble corresponds to the \textbf{\textcolor{purple}{dual amplitude}} of string theory, defined on a discrete target space $\vec F(\theta)$ with distributed external momentum $\vec Q(\theta,t)= \frac{\vec{C}'(\theta)}{\sqrt{2 \nu (t + t_0)}}$:
        $$
            \Psi[C, t] = \VEV{\exp{ \I \oint d\theta \vec{F}(\theta) \cdot \vec{Q}(\theta,t)}}_{F}$$
        \item Averaging over \textbf{\textcolor{teal}{string target space}} $\vec F_k$ corresponds to summing over star polygons with unit sides and rational angles $\beta = 2\pi \frac{p}{q}$.
        \item Averaging over \textbf{\textcolor{teal}{fermionic/Ising degrees of freedom}} $ \sigma_k$ produces a \textcolor{purple}{random walk} (Brownian motion in the continuum limit) across polygon vertices.
        \item The viscosity enters this string theory as a coupling constant in the denominator of the effective Action. The turbulent limit of $\nu \to 0$ becomes the weak coupling limit, solvable in the WKB approximation.
\end{itemize}

Using number theory methods, we compute statistics of these radii, which leads to the solution of this string theory in the WKB limit.

The WKB limit $\nu \to 0$ corresponds to the strong turbulence in the original NS equation. 

Thus, we have found a dual string theory for the NS equation, with the weak coupling corresponding to decaying turbulence.

There are many technical details, including the WKB computations\cite{migdal2024quantum} (instanton in this string theory) and initial data for the loop functional corresponding to the Gibbs distribution.

However, the essence of the theory is summarized in the above text. Any modern QFT/String theorist can reproduce the rest (perhaps, with some help from number theory colleagues).

Before reading the remaining sections of this paper, I recommend browsing the slides \cite{RutgersSlides2025}, summarizing the essence of the theory, and its experimental verification.

\section{Introduction}

Turbulence remains one of the most challenging open problems in classical physics. 
Traditional treatments rely heavily on phenomenological models and heuristic assumptions, 
such as Kolmogorov's scaling (K41) and multifractal frameworks. Despite extensive 
efforts, the fundamental question of global regularity for the Navier–Stokes (NS) equations, 
recognized as one of the Millennium Prize problems, remains unresolved.

In this paper, we introduce a novel theoretical approach, transforming the NS equations 
into a mathematically more tractable loop-space formulation:
\begin{itemize}
    \item We explicitly reformulate the three-dimensional NS equations in terms of the 
    momentum loop \(\vec{P}(\theta, t)\), converting the original three-dimensional PDE 
    into a solvable one-dimensional nonlinear equation.
    
    \item From this reformulation, we derive the \textit{No Explosion Theorem}, which 
    proves that finite-time singularities do not occur under stochastic 
    initial conditions due to thermal fluctuations of velocity field.
    
    \item We introduce the \textit{Euler ensemble}, an exact analytical solution 
    describing the universal asymptotic state of decaying turbulence, strongly supported 
    by numerical simulations and experimental observations.
\end{itemize}
\section{Loop functional and its general properties}\label{loopFunc}
The loop functional is defined as a phase factor associated with velocity circulation, averaged over the initial distribution $\vec v_0$ of the velocity field
\begin{eqnarray}
  &&  \Psi(\gamma, C) = \VEV{\exp{\frac{\I \gamma}{\nu}\Gamma }}_{\vec v_0};\\
  && \Gamma = \oint_C \vec v(\vec r) \cdot d \vec r;
\end{eqnarray}
We use viscosity $\nu$ as a unit of circulation. Both have the same dimension $L^2/T$ as the Planck's constant $\hbar$.
The viscosity will play the same role in our theory as Planck's constant in quantum mechanics.
The variable $\gamma$ with this definition is dimensionless.

This loop functional is the Fourier transform of the PDF for the circulation over fixed loop $C$
\begin{eqnarray}
    && \Psi(\gamma, C) = \int_{-\8}^\8 d \Gamma P(\Gamma, C) \exp{ \frac{\I \gamma}{\nu} \Gamma };\\
    && P(\Gamma, C) = \int_{-\8}^\8 \frac{d \gamma}{2 \pi\nu} \Psi(\gamma, C) \exp{ -\frac{\I \gamma}{\nu}\Gamma }
\end{eqnarray}
There is an implicit dependence of time, coming from the evolution of the velocity field by the \NS{} equation
 \begin{eqnarray}\label{NSEQ}
    &&\partial_t \vec v = -\nu \nabla \times \vec \omega - \vec v\times \vec \omega - \vec \nabla \left( p + \frac{\vec v^2}{2} \right);\\
    && \nabla \cdot \vec v =0;\\
    &&\vec \omega = \vec \nabla \times \vec v
\end{eqnarray}
We restrict ourselves to three-dimensional Euclidean space, the most interesting case for physics applications. The generalization to arbitrary dimension is straightforward, as discussed in previous papers \cite{M93, M23PR, migdal2023exact}.

In the next sections, we shall study the Cauchy problem for the loop equation \cite{M93, M23PR}, which follows from the \NS{} equation.
Here, we state some general properties of the loop function and various scenarios of its evolution.
\pctPDF{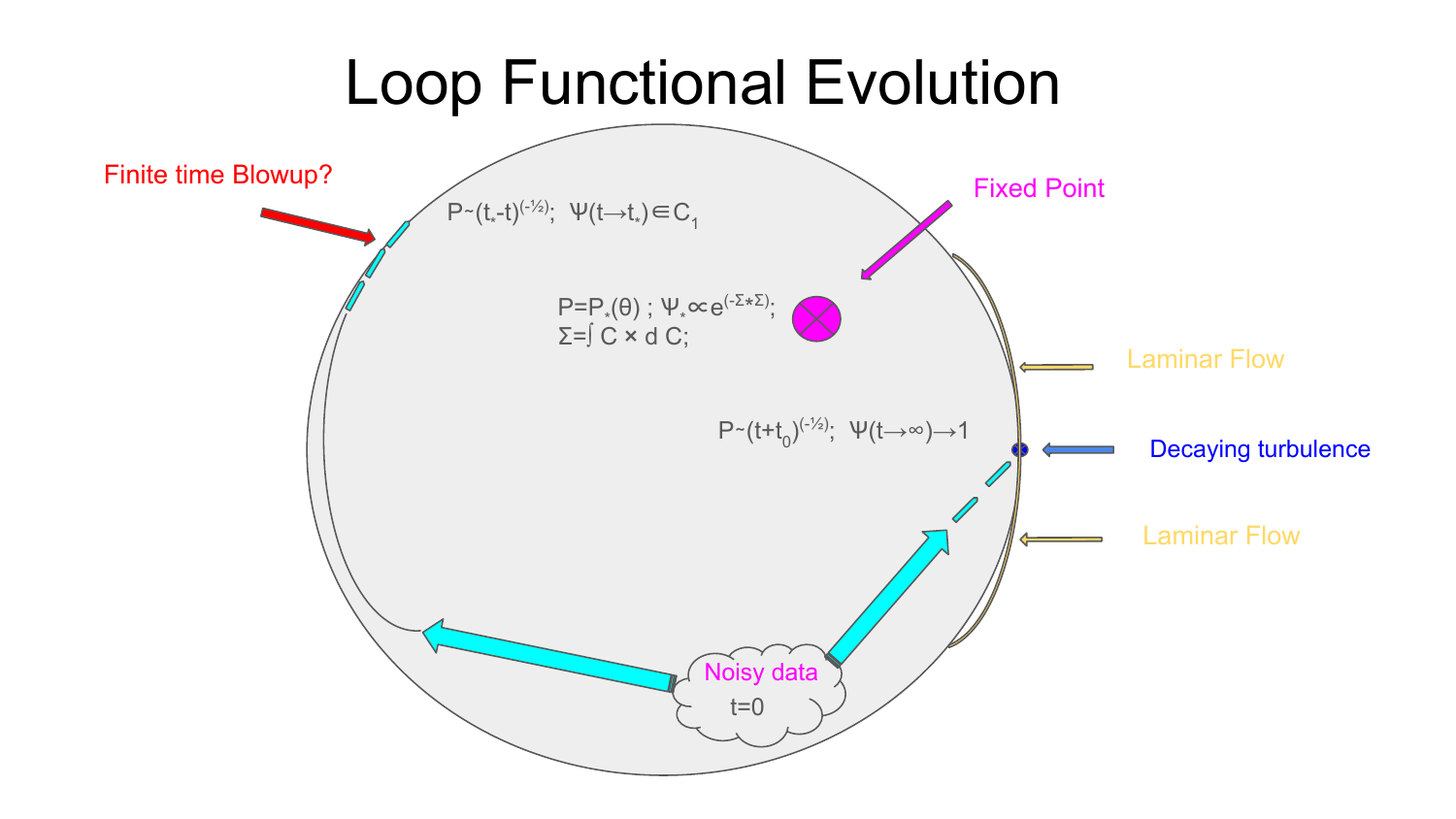}{Asymptotic trajectories of the time evolution for the loop functional inside the unit circle in the complex plane. The laminar flow is the yellow region on the circle close to $\Psi =1$. Three other flows are 1) hypothetical explosion, 2) decaying turbulence, and 3) special fixed point.}
The first obvious property is that this evolution goes inside the unit circle
\begin{eqnarray}
    && \abs{\Psi(\gamma, C)}\le 1 ; \forall t;
\end{eqnarray}
At a small enough time passed from initial data, $ t < t_c$, turning off the noise would bring us to the usual unique laminar solution of the \NS{} equation, corresponding to the loop functional at the unit circle with a small enough phase.
\begin{eqnarray}
    && \lim_{\sigma\to 0}\abs{\Psi(\gamma, C)} = 1;\forall t < t_c; 
\end{eqnarray}
Here, $\sigma $ denotes the variance of the Gaussian distribution of the velocity field around some smooth initial value.
Generally speaking, we could expect the following fixed points (see Fig.\ref{fig::LoopFunctionalEvolution}) of the time evolution for the loop functional\footnote{We do not count deterministic fixed points corresponding to potential flows. They correspond to isolated points moving on the unit circle.}.

\begin{enumerate}
     \item \textbf{Special solution}. There is a fixed point corresponding to the global random rotation of the fluid (see \cite{M93, M23PR, migdal2023exact}).
    \item \textbf{Decaying Turbulence}. The evolution of loop average reaches some \textbf{fixed trajectory}, independent of initial data, and covers some nontrivial manifold (see \cite{migdal2023exact, migdal2024quantum}). At infinite time, this fixed trajectory leads to zero velocity, corresponding to all the kinetic energy dissipated by viscous effects.
    \item \textbf{Finite-time explosion?} The vorticity could blow up at some finite or infinite point in time, leading to infinite circulation. 
   
\end{enumerate}

In the following sections, we elaborate on each of these possible regimes. The finite-time explosion is proven to be inconsistent and is therefore ruled out.
\section{Loop equation}\label{HopfLoop}
The first step is to write down the loop equation by projecting the Hopf equation \cite{Hopf19} to the loop space.

Before doing that, we have to specify certain boundary conditions which we assume in our fluid dynamics. Namely, we consider infinite space $\mathbb R_3$, with boundary condition of vanishing or constant velocity at infinity.

Vorticity can exist throughout the entire spatial domain, but it vanishes at infinity, as required by our boundary conditions, where velocity gradients diminish. We also assume the absence of internal boundaries, such as solid surfaces that the fluid flows around. However, we allow for the presence of singular regions of vorticity, such as vortex lines and sheets, in the inviscid limit, provided these regions are confined to a finite portion of the volume, consistent with our boundary conditions.

At finite viscosity, these singular regions may acquire a finite thickness proportional to $\sqrt{\nu}$, leading to anomalous dissipation in turbulent flows.

An alternative, simpler mechanism for anomalous dissipation arises when velocity increases as viscosity vanishes, even in the absence of singular vortex structures like Burgers vortices. In this scenario, the turbulent energy grows as viscosity decreases, while the spatial distribution of vorticity remains homogeneous and non-singular.

As we will demonstrate below, this is the mechanism realized in our solution for decaying turbulence.

The computations leading to the loop equation were performed in the old papers \cite{M93, M23PR}. For the reader's convenience, we repeat them here using another language, hopefully more clear for mathematicians.

The straightforward time derivative of the loop functional, assuming the constant loop $C$ and using time derivative \eqref{NSEQ} of the velocity field in the circulation, yields
\begin{eqnarray}\label{Tder}
  &&  \partial_t \Psi(\gamma,\vec C) =\VEV{\frac{\I \gamma}{\nu}\oint d \vec C(\theta) \cdot \vec L(\vec C(\theta))\exp{ \frac{\I \gamma}{\nu}\Gamma(\vec v, \vec C)}}_{sol};\\
  && \vec L(\vec r) =  - \nu \vec \nabla \times \vec \omega(\vec r) + \vec \omega(\vec r)\times\vec v(\vec r)
\end{eqnarray}
The averaging $\VEV{}_{sol}$  goes, as before, over time-dependent NS solutions $\vec v(\vec r)$ with a given set of initial values $\vec v_0(\vec r)$. It is implied that a probability measure (see examples below) is supplied for this set of initial velocity fields. The phase factor of circulation is averaged over initial data using this measure.

The gradient terms $\vec \nabla \left( p + \frac{\vec v^2}{2} \right)$ in \eqref{NSEQ} dropped in the time derivative of the circulation as the integral of a gradient of some single-valued function of coordinate $H(\vec r) = p(\vec r) +  \frac{\vec v^2(\vec r)}{2} $ around the closed loop: $$ \oint d \vec C(\theta) \cdot \vec \nabla H(\vec C(\theta)) = \oint d H(\vec C(\theta)) =0.$$
The velocity field $\vec v $ is a solution of the Poisson equation, relating it to vorticity by incompressibility condition
\begin{eqnarray}\label{VOM}
&&\vec v(\vec r) = \frac{-1}{\vec \nabla^2}\vec\nabla\times \vec \omega(\vec r)
\end{eqnarray}
This representation leaves vorticity as the main unknown variable in the time derivative of the loop functional.

To find the loop equation, we must replace the vorticity and its gradients with certain operators acting on the loop independently of the vorticity and velocity fields.
As a result of such transformation, the vector function $\vec L(\vec C(\theta))$ will be replaced by a certain operator $\hat L(\theta)$ in loop space acting on $\Psi(\gamma,\vec C)$
\begin{eqnarray}\label{Lop}
    && \hat L(\theta) \exp{ \frac{\I \gamma}{\nu}\Gamma(\vec v, \vec C)} =\left(  - \nu \hat \nabla(\theta)  \times \hat\omega(\theta) + \hat \omega(\theta) \times \hat v(\theta)\right)\exp{ \frac{\I \gamma}{\nu}\Gamma(\vec v, \vec C)};\\
    \label{lopEq}
    &&  \partial_t \Psi(\gamma,\vec C) =\frac{\I \gamma}{\nu}\oint d \vec C(\theta) \cdot \hat L(\theta)\Psi(\gamma,\vec C)
\end{eqnarray}
This operator $\hat L(\theta)$ depends on certain operators $\hat \nabla(\theta),  \hat\omega(\theta),  \hat v(\theta)$ instead of on the dynamical variables $\vec v(\vec r), \vec \omega(\vec r)$, therefore it can be taken out of the averaging over trajectories starting from various initial data $\vec v_0(\vec r) $ so that this operator acts on the loop average $\Psi(\gamma,\vec C)$.
Such is the plan of the proof of the loop equation.
We define the loop operators and follow this plan in the next section.

\section{The definitions of the loop operators and the proof of the loop equation}\label{LoopDef}
The operators in the loop equation were introduced in\cite{M93} and explained at length in my review paper\cite{M23PR}.

In this paper, we do not assume any knowledge of the previous work; instead, we derive the loop operators from scratch using a simpler method.

First, we approximate the smooth loop $C(\theta)$ by a polygon with $N$ vertices $\vec C_k = \vec C(2 \pi k/N)$ in the limit  $ N \to \8$. We postpone this local limit until we solve the discrete loop equation. This limit will define the continuum theory in the same way as in the QFT; the functional integral is discretized using a lattice with the lattice spacing going to zero at the end of the calculation. In this limit, the theory's parameters vary with the lattice spacing to provide a finite result for the physical observables.

The first observation is that with smooth velocity and vorticity fields, the discrete circulation around the polygon $\vec C_k, k =0,\dots N-1, \vec C_N = \vec C_0$ converges to the circulation around the smooth loop
\begin{eqnarray}
    &&\Gamma \equiv \sum_k \int_{C_k}^{C_{k+1}} d \vec r \cdot \vec v(\vec r, t) \to \oint d \vec C(\theta) \cdot \vec v(\vec C(\theta));
\end{eqnarray}
The line integral $$\int_{C_k}^{C_{k+1}} d \vec r$$ becomes the continuous integral along the loop $C$ in the limit $N\to \infty$ when the sides $\Delta \vec C_k = \vec C_{k+1} - \vec C_k$ of the polygon tend to zero.

The next property is also easy to prove using the Stokes theorem for a small triangle $\left(\vec C_{k-1}, \vec C_k, \vec C_{k+1}\right)$
\begin{eqnarray}
&& \vec \nabla_k \equiv \pd{\vec C_k};\\
    && \vec \nabla_k \Gamma \propto \left( \Delta \vec C_k + \Delta \vec C_{k-1} \right)\times \vec \omega(\vec C_k)  \to 0
\end{eqnarray}

To prove this we rewrite the complete circulation as a sum of two:
\begin{eqnarray}
    &&\Gamma = \Gamma' + \Gamma_\triangle;
\end{eqnarray}
The circulation $\Gamma'$ corresponds to the loop with one vertex $C_k$ removed, resulting in a side connecting $C_{k-1}$ directly to $C_{k+1}$. The other term $\Gamma_\triangle$ corresponds to circulation around the triangle made of three vertices $\triangle = \left(\vec C_{k-1}, \vec C_k, \vec C_{k+1}\right)$.  When these two circulations are added, the line integral $\int_{C_{k-1}}^{C_{k+1}} d \vec r \cdot \vec v(\vec r, t)$ cancels between the two, so we are left with original circulation over complete polygon $C$.

The derivatives by $\vec \nabla_k \equiv \pd{\vec C_k}$ only involve the triangle circulation $\Gamma_\triangle$ and are easily estimated at $N \to \infty$.

This first derivative vanishes at $N \to \8$ as $\Delta \vec C_k \sim \mathcal O(1/N)$.

The second derivative, however, stays finite. We prefer to use another set of variables
\begin{eqnarray}
   && \vec s_k = \Delta \vec C_k; \\
   &&\vec \eta_k = \pd{\vec s_k};\\
   && \vec \nabla_k = -\Delta\vec \eta_{k-1};
\end{eqnarray}
The last relation follows from the chain rule
\begin{eqnarray}
    &&\vec \nabla_k = \frac{\partial{\vec s_k}}{\partial{\vec C_k} }\cdot \vec\eta_k +  \frac{\partial{\vec s_{k-1}}}{\partial{\vec C_k}}\cdot \vec \eta_{k-1} =\vec \eta_{k-1} - \vec \eta_k  = - \Delta \vec \eta_{k-1}
\end{eqnarray}
The vorticity can be represented as
\begin{eqnarray}
    && \vec \eta_{k-}\times \vec \nabla_{k} \Gamma \to \vec \omega(\vec C_k) + \mathcal O(1/N);\\
    && \vec\eta_{k-} \equiv \frac{\vec \eta_{k}+ \vec \eta_{k-1}}{2};
\end{eqnarray}

The contour $C$ becomes an open line when we move all $\vec s_k$ independently, without restricting $\sum \vec s_k =0$. However, the contribution to the time derivative of circulation from the extra gap between the endpoints $\Delta \partial_t \Gamma \propto H(\vec C_N) - H(\vec C_0) $ where $H(\vec r) = p(\vec r) + \frac{\vec v^2(\vec r)}{2}$ is the enthalpy, which is supposed to be differentiable. Thus, this error term goes to zero as we reinstate the closure condition $\sum \vec s_k =\vec C_N - \vec C_0 =0$.

Finally, the velocity field at the vertex $\vec v(\vec C_k)$ can be related to the vorticity through the Biot-Savart law
\begin{eqnarray}
   && \vec v(\vec C_k) \exp{\frac{\I \gamma\Gamma} {\nu}}= -1/(\vec \nabla_k^2) \vec \nabla_k \times \vec \omega(\vec C_k)  \exp{\frac{\I \gamma\Gamma} {\nu}};
\end{eqnarray}
Let us verify this relation using the Biot-Savart integral formula for the inverse Laplace operator
\begin{eqnarray}
    &&\vec v(\vec C_k)\exp{\I \Gamma}=\frac{1}{4 \pi} \int d^3 r \frac{\vec r \times \vec\omega(\vec C_k+ \vec r)}{|\vec r|^3}\exp{\I \tilde \Gamma(\vec r)} +  \mathcal O(1/N);\\
    && \tilde \Gamma(\vec r) = \left. \Gamma\right|_{\vec C_k \Rightarrow \vec C_k + \vec r}
\end{eqnarray}
At first glance, the loop in the new circulation $\tilde \Gamma(\vec r)$ involves two long "wires": $(\vec C_{k-1}, \vec C_k+ \vec r)$ and $(\vec C_k + \vec r, \vec C_{k+1})$. 

However, in the local limit, when the distance $|\vec C_{k+1}- \vec C_{k-1}| = \mathcal O(1/N) $, these two wires have zero area inside the arising thin triangle, so they effectively cancel in virtue of the Stokes theorem, assuming the Biot-Savart integral converges.
\begin{eqnarray}
    \tilde \Gamma(\vec r) \to \tilde \Gamma(0) = \Gamma
\end{eqnarray}
This produces the desired result in the Biot-Savart formula.

The convergence of the \BS{} integral follows from our boundary conditions, assuming no vorticity at infinity or even stronger requirement of finite support of vorticity.
The phase factor $\exp{\I \tilde \Gamma(r)} $ does not influence the absolute convergence, so it can be set to $\exp{\frac{\I \gamma\Gamma} {\nu}}$ for that purpose and taken out of the integral, returning us to the convergence of the ordinary \BS{} integral.

Therefore, with $\mathcal O(1/N)$  accuracy, we can replace the right side of the \eqref{Tder} by its discrete version with operators involving $\vec \nabla_k$
\begin{subequations}\label{LoopEq}
    \begin{eqnarray}
    &&\pd{t}\VEV{\exp{\frac{\I \gamma\Gamma} {\nu}}}= \frac{\I \gamma}{\nu}\sum_k \Delta \vec C_k \cdot \hat L_k \VEV{\exp{ \frac{\imath  \gamma\Gamma}{\nu}}} + \mathcal O(1/N);\\
    && \hat L_k = - \nu \vec \nabla_k \times \hat \omega_k + \hat \omega_k \times \hat v_k  ;\\
    && \hat v_k = -1/(\vec \nabla_k^2) \vec \nabla_k \times \hat \omega_k;\\
    && \hat \omega_k = \frac{\I\gamma}{\nu}\vec \eta_{k-} \times \vec \nabla_{k};
\end{eqnarray}
\end{subequations}
We restrict ourselves to the velocity vanishing at infinity and no internal boundaries in the physical domain.
With this boundary condition, the harmonic potential is zero, and there is no zero mode to add to the inverse Laplace operator.

\textbf{In the rest of the paper, we shall use the language of the continuum theory, implying the limit $N\to \8$ of a polygon  $\vec C$ with $N$ sides. While the lengths of the sides of $\vec C$ vanish in the local limit $N \to \8$, the sides of $\vec P$ polygon are not at our disposal, so they may stay finite (this will happen in the decaying turbulence below).}
\section{Schrödinger equation in loop space}\label{HPsiLoop}
Before we investigate the solutions of the loop equation, let us consider its physical and mathematical meaning and its relation to the geometry of the incompressible flow.

By definition, the loop functional $\Psi(\gamma,\vec C)$ is a superposition of the phase factors $\exp{ \frac{\I \gamma}{\nu}\Gamma(\vec v, \vec C)}$ with the circulation $\Gamma$ of a particular solution $\vec v(\vec r, t)$ of the \NS{} equation. These solutions have initial values $\vec v(\vec r, 0) = \vec v_0(\vec r)$, distributed by some distribution $P[\vec v]$ which we assume Gaussian with the mean given by some smooth initial field and some coordinate-independent variance $\sigma$.

In the turbulent scenario, the \NS{} trajectories initiated from a narrow vicinity of some smooth velocity field eventually expand and cover some attractor, slowly varying with time and asymptotically converging to $\vec v =0, \Psi =1$.

The alternative smooth solution of the \NS{} equation, sought after in numerous mathematical papers, would correspond to these trajectories staying close and converging to a single trajectory in the limit $\sigma\to 0$. This single trajectory would go along the unit circle, bounding our disk.

With this generalization of a definition of the Cauchy problem for the \NS{} equation, we can address the existence of smooth, explosive, or stochastic (i.e., turbulent) solutions within the loop equation's framework.

The transformation from the \NS{} equation to the loop equation is similar to that from the Newton equation of the particle in random media to the diffusion equation. We add dimension to the problem, switching to the probability distribution in $\mathbb{R}_d$, after which the particle's infinitesimal time steps translate into probability derivatives by coordinates.

There are two essential differences, however. Our loop space is not just higher-dimensional; in the local limit $N \to \8$, it is infinite-dimensional.
The second difference is that in addition to diffusion terms $ \nu \hat \nabla \times \hat \omega$, we have nonlocal advection terms $\hat v \times \hat \omega $ affecting the evolution of the distribution in loop space.

Our definition of the loop functional already by construction has superficial similarities with quantum mechanics. We are summing phase factor over a manifold of solutions of the \NS{} equations. The circulation plays the role of classical Action, and viscosity plays the role of Planck's constant.

This analogy becomes a complete equivalence when the time derivative of the loop functional is represented as an operator $\vec L(
\vec C(\theta)) \Rightarrow \hat L(\theta)$ in the loop space acting on this functional.

Now we have quantum mechanics in loop space, with the Hamiltonian $$\hat H \propto \oint d \vec C(\theta) \cdot \hat L(\theta).$$ 
The operator $\hat L(\theta)$ depends of functional derivatives $\ff{\vec C(\theta)}$, as was determined, and discussed in previous works \cite{M93, M23PR, migdal2023exact}. Our polygonal approximation has no functional derivatives, just ordinary derivatives $\vec \nabla_k = \pd{\vec C_k}, \vec \eta_k = \pd{\Delta \vec C_k}$.  Thus, our quantum-mechanical system has $3 N$ continuum degrees of freedom $\vec C_1, \dots \vec C_N$ with periodicity constraint $ \vec C_0 = \vec C_N$.

This Hamiltonian is not Hermitian, which reflects the dissipation phenomena. The time reversal leads to complex conjugation of the loop functional, a nontrivial transformation, as there is no symmetry for the reflection of velocity field $\vec v(\vec r,t) \to -\vec v(\vec r, t)$.
 
The loop in our theory is a periodic function of the angular variable $\theta$. Geometrically, this is a map of the unit circle into Euclidean space  $\mathbb{S}_1 \mapsto \mathbb R_d$. 
In particular, there could be several smaller periods, in which case this loop becomes a set of several closed loops connected by backtracking wires like in Fig.\ref{fig::VPsi}.
\pctPDF{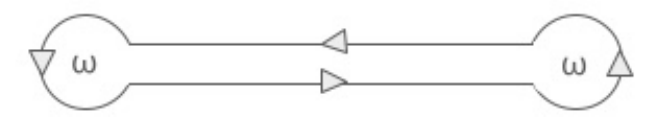}{The "hairpin" loop $C$ used in defining the pair correlation of vorticity. The little circles are the loop variations needed to bring down vorticity at two points in space. The backtracking contribution to the circulation cancels at vanishing separation between these parallel lines.}
Also, this map could have an arbitrary winding number $n$ corresponding to the same geometric loop in $\mathbb{R}_d$ traversed $n $ times.

The linearity of the loop equation is the most important property of this transformation from \NS{} equation to the quantum mechanics in loop space.

This transformation exemplifies how the nonlinear PDE reduces to the linear problem projected from high dimensional space. In our case, this space is the loop space, which is infinite-dimensional. 

As a consequence of linearity, the generic solution of the loop equation is a superposition of particular solutions with various parameters. More generally, this is an integral (or sum, in discrete case) over the space $\mathcal S$ of solutions of the loop equation.

In the case of the Cauchy problem in loop space, the measure for this integration over space $\mathcal S$ is determined by the initial distribution of the velocity field.
The asymptotic turbulent solution \cite{migdal2023exact} uniformly covers the Euler ensemble, like the microcanonical distribution in Newton's mechanics covers the energy surface.

This turbulent solution does not solve a Cauchy problem; it rather solves the loop equation with the boundary condition at infinite time $ \Psi_{t= \8} =1 $.

In the next section, we simplify the loop equation using Fourier space; this will be the foundation for the subsequent analysis.
\section{Momentum Loop Equation}\label{MomLoopEq}

The loop operator, $\hat L$ in \eqref{LoopEq}, dramatically simplifies in the functional Fourier space, which we call momentum loop space. In our discrete approximation, the momentum loop will also be a polygon with $N$ sides.

The origin of this simplification is the lack of direct dependence of the loop operator $\hat L(\theta)$ on the loop $C$ itself. Only derivatives $\vec \nabla_k, \vec \eta_k$ enter this operator.

From the point of view of quantum mechanics in loop space, our Hamiltonian only depends on the canonical momenta but not on the canonical coordinates. This property is exact as long as we do not add external forces.

This remarkable symmetry property (translational invariance in loop space) allows us to look for the "superposition of plane waves" Ansatz:
\begin{subequations}
    \label{PsiP}
\begin{eqnarray}
    &&\Psi(\gamma,C|t) =  \VEV{\psi_p(t)}_{init};\\
    &&\psi_p(t) =  \exp{\frac{\I \gamma}{\nu}\sum_k \Delta \vec C_k \cdot \vec P_{k}(t)}
\end{eqnarray}
\end{subequations}
Here the averaging $\VEV{\dots}_{init}$ goes over all trajectories $P_{k}(t)$ passing through random initial data $\vec P_k(0)$ distributed with the corresponding probability to reproduce initial value $\Psi(\gamma,C|0)$. We discuss this initial distribution in the next sections.

The operators $\vec \nabla_k, \vec \eta_k$ become ordinary vectors when applied to $\psi_p$ in \eqref{PsiP}:
\begin{subequations}\label{nablaEta}
  \begin{eqnarray}
    &&\vec \nabla_k \psi_p = -\frac{\I \gamma}{\nu}\Delta \vec P_{k-1} \psi_p;\\
    &&\vec \eta_{k-} \psi_p = \frac{\I \gamma}{\nu}\vec P_{k-} \psi_p;\\
    && \vec P_{k-}  \equiv \frac{\vec P_k + \vec P_{k-1}}{2};\\
    \label{omegaK}
    && \hat \omega_k \propto \frac{\I \gamma}{\nu}\vec P_{k-} \times\Delta \vec P_k
\end{eqnarray}  
\end{subequations}
The velocity circulation can be rewritten up to $\mathcal O(1/N)$ corrections as a symmetric sum
\begin{eqnarray}\label{sumKminus}
   &&\sum_k \Delta \vec C_k \cdot \vec P_{k}(t) + \mathcal O(1/N) =  \sum_k \frac{\Delta \vec C_k + \Delta C_{k+1}}{2} \cdot \vec P_{k}(t) =\sum_k \Delta \vec C_k \cdot \vec P_{k-}(t) 
\end{eqnarray}
We did not assume here anything about the continuity of $\vec P_k$; we only assumed that $|\Delta C_{k+1} - \Delta \vec C_k|  \ll |\Delta \vec C_k| $ which is true for smooth loop.

The discrete loop equation \eqref{LoopEq} with our Anzatz \eqref{PsiP} after some algebraic transformations using the above identities \eqref{nablaEta}, \eqref{sumKminus} reduces to
the following momentum loop equation (MLE)\cite{M23PR,migdal2023exact}
\begin{subequations}\label{PloopEq}
   \begin{eqnarray}
  && \nu\partial_t \vec{P} =  - \gamma^2(\Delta \vec{P})^2 \vec{P}  + \Delta \vec{P} \left(\gamma^2 \vec{P} \cdot \Delta \vec{P} +\I \gamma \left( \frac{(\vec{P} \cdot \Delta \vec{P})^2}{\Delta \vec{P}^2}- \vec{P}^2\right)\right);\\
  && \Delta \vec P \equiv \vec P_{k} - \vec P_{k-1};\\
  && \vec P \equiv \vec P_{k-}
\end{eqnarray}
\end{subequations}
 
In the local limit $N \to \8$, the momentum loop will have a discontinuity $\Delta \vec{P}(\theta)$ at every parameter $0<\theta \le 2 \pi$, making it a fractal curve in complex space $\mathbb{C}_d$.  Such a curve can only be defined using a limit like a polygonal approximation or a Fourier expansion of a periodic function of $\theta$ with slowly decreasing Fourier coefficients. 

You can regard this curve as a periodic random process hopping around the circle (more about this process below, in the context of the decaying turbulence).

The details can be found in \cite{M23PR, migdal2023exact}. We will skip the arguments $t, k$ in these loop equations, as there is no explicit dependence of these equations on either of these parameters.

\section{Uniform constant rotation and momentum loop}\label{RandRot}
The loop equation has several unusual features, especially the discontinuities of the momentum loop. These discontinuities have a physical meaning related to vorticity.

It is best understood by studying an exact fixed point of the loop equation: the global constant rotation. We set $\gamma = \nu$ for simplicity in this example.
\begin{eqnarray}\label{vrot}
    && \val(\vec r|\phi) = \phi_{\alpha\beta} \rbe;\\
    && \phi_{\alpha\beta} = - \phi_{\beta\alpha};\\
    && \Psi[C] = \exp{\I  \phi_{\alpha\beta} \oint d C_\alpha(\theta) C_\beta(\theta)};
\end{eqnarray}
We present two implementations of the momentum loop for this simple model: one using an infinite Fourier expansion and another using the limit of polygonal approximation of the loop. This will allow us better understand the origin and the meaning of these discontinuities.

\subsection{Infinite Fourier series}
Here is the implementation of the momentum loop by an infinite Fourier series.
\begin{eqnarray}
\label{Pexp}
&&\Psi_0[C]=  \VEV{\exp{
	   \I \oint d \vec{C}(\theta) \cdot \vec{P}(\theta)}}_{P};\\
  &&  P_\alpha(\theta)= \sum_{\text{odd }n=1}^\infty P_{\alpha,n} e^{\I n \theta} + \bar{P}_{\alpha,n} e^{-\I n \theta};\\
  && P_{\alpha,n} = \mathcal N(0,1) ;\\
  && \bar{P}_{\alpha,n} =\frac{4 }{\pi n } \phi_{\alpha \beta}P_{\beta,n} ;\\
  && \phi_{ \alpha\beta} = - \phi_{\beta\alpha};
\end{eqnarray}
The covariance matrix components are (for odd $n,l$)
\begin{eqnarray}
  &&\VEV{P_{\alpha,n} P_{\beta,l}}
= \frac{4 }{n } \delta_{n l} \phi_{\alpha \beta};\\
\label{Pcorr}
&&\VEV{ P_{\alpha}(\theta) P_{\beta}(\theta')}_P =
2\imath  \phi_{\alpha \beta} \sign(\theta'-\theta);
\end{eqnarray}
The loop functional is obtained after averaging over Gaussian random variables $ P_{\alpha,n}, \phi_{\alpha\beta}$. The loop function can be computed without an explicit Functional Fourier transform using the well-known properties of the Gaussian expectation value of the exponential.
\begin{eqnarray}
   && \VEV{\exp{
	  \I \oint d \vec{C}(\theta) \cdot\vec{P}(\theta)}}_{P} \nonumber\\
    &&\propto \exp{
	   -\oh\oint d C_\alpha(\theta) \oint d C_\beta(\theta')\VEV{P_\alpha(\theta) P_\beta(\theta')} }\nonumber\\
    &&\propto \exp{
	-\imath/2 \phi_{\alpha \beta}\oint d C_\alpha(\theta) \oint d C_\beta (\theta') \sign(\theta-\theta')} =\nonumber\\
    && \exp{
	   -\imath \phi_{\alpha \beta}\Sigma _{\alpha\beta}};\\
    && \Sigma_{ \alpha\beta} = \oint d C_\alpha(\theta) C_\beta (\theta)
\end{eqnarray}

With this representation, it is obvious why the circulation does not depend on time; the vorticity is a global constant $\phi_{\alpha\beta}$ which does not depend on time nor $\vec r$. Simple tensor algebra in the time derivative of circulation leads to the term
\begin{eqnarray}
   && \oint d C_\alpha(\theta)   L_\alpha(\theta) \propto \phi_{\alpha\beta}\phi_{\beta\gamma} \Sigma_{\gamma\alpha} =0;\\
   && \Sigma_{\alpha\gamma} = -\Sigma_{\gamma\alpha} =\oint_C r_\alpha d r_\gamma
\end{eqnarray}
The tensor trace vanishes by symmetry $\gamma \leftrightarrow \alpha$, changing the sign of $\Sigma_{\alpha\gamma}$. This solution is a consequence of the rotational symmetry of the \NS{} equation.

Verification of the MLE is more tedious because, this time, the velocity in \eqref{vrot} explicitly depends on the coordinate. This will become $ \phi_{\alpha\beta} C_\beta(\theta)$ in the equation, which means that the operator $\hat L(\theta)$ depends both on $\vec C, \vec P$. 
Still, for the proof, it suffices to know the momentum loop \eqref{PsiP} and the corresponding velocity field \eqref{vrot}, solving the \NS{} equation for arbitrary constant $\phi$.

Though this special solution does not describe isotropic turbulence, it helps understand the mathematical properties of the loop technology. 

In particular, it shows the significance of the discontinuities of the momentum loop $\vec{P}(\theta)$, as it is manifest in the correlation function\eqref{Pcorr}. These discontinuities are necessary for vorticity; they result from the divergence of the Fourier series in \eqref{Pexp}. 

The mean vorticity at the circle is proportional to $\phi_{\alpha \beta}$ independently of $\theta$
\begin{eqnarray}
    &&\VEV{\omega_{\alpha\beta}(\vec C(\theta))} \propto \VEV{ P_{\alpha}(\theta)
    \Delta P_{\beta}(\theta)} \propto   \phi_{\alpha \beta} ;
\end{eqnarray}
\subsection{Polygonal approximation}
The second implementation is more aligned with the methods we use in the MLE. We approximate the loop $C$ as a polygon with vertices $\vec C_k$ equidistant on a parametric circle.
\begin{eqnarray}\label{PsiDisc}
    && \I\int_C  \vec C(\theta) \cdot \phi \cdot  d\vec C(\theta) \approx \I\sum_{k=0}^{N-1}\vec C(k) \cdot \phi \cdot \Delta C(k) =\nonumber\\
    && \frac{\I}{2}\sum_{k,l=0}^{N-1} \Delta \vec C(k) \cdot \phi \cdot \Delta \vec C(l) \sign(k-l);\\
    && \Delta \vec C(k) = \vec C(k+1) - \vec C(k);
\end{eqnarray}
Our next task is to represent the loop functional as a Gaussian average over the momentum loop $\vec P= \{\vec P_0, \dots \vec P_k$\} with symmetric covariance matrix
\begin{eqnarray}
    \VEV{P^\alpha_k P^\beta_l} =  \I \phi_{\alpha\beta} \sign(k-l)
\end{eqnarray}
This representation will involve the following discrete Fourier transform with Gaussian coefficients
\begin{eqnarray}
  &&  P_k^\alpha = \sum_{n=1}^{N} \xi^\alpha_n \exp{\I k \omega_n} + \bar\xi^\alpha_n \exp{-\I k \omega_n};\\
  &&\omega_n = \pi (2 n+1)/N;\\
  && \VEV{\xi^\alpha_n \bar\xi^\beta_m} = \I\phi_{\alpha\beta} \delta_{n , m} U(n);\\
  && U(n) =\frac{2}{N}\sum_{k= -N}^{N}\sign(k) \sin\left(k \omega_n\right) 
\end{eqnarray}
This  discrete Fourier transform for $U(n)$ reduces to a finite geometric series with the following result
\begin{eqnarray}
    U(n) = \frac{2}{ N \tan\left(\frac{\omega_n}{2}\right)};
\end{eqnarray}
Note that this $\vec P_k$ is antiperiodic: it changes the sign when the index goes around the loop. This, however, keeps the solution simply periodic in $C$ space, as only an even number of $\vec P$ variables have non-vanishing expectation values in this particular example.

This example shows both the discontinuities' meaning and the momentum loop's approximation by a polygon.
In this example, the continuum limit $N \to \infty$ can be taken for the loop functional, but not at the level of the Fourier series for the momentum loop. 

The formal limit $N \to \infty$ exists for $U(n)$ at fixed $n$
\begin{eqnarray}
    U(n)_{N \to \8} \to \frac{4}{\pi(2 n +1)}
\end{eqnarray}
and matches the continuum theory,
but the oscillating sum of Gaussian random variables does not converge to any normal function; rather, this is a stochastic process on $\bS_1$ with convergent expectation values.





\section{Cauchy problem and its solution}\label{Cauchy}

The Cauchy problem, notoriously difficult for nonlinear PDE, can be solved analytically for the loop equation. The hard part of the problem is now hidden in the limit $\sigma \to 0$, bringing us back to the \NS{} equation with smooth initial data.

Let us describe this solution. Assuming the MLE \eqref{PloopEq} satisfied, we have certain conditions for the initial data $\vec P_0(\theta) = \vec P(\theta, 0)$. This data is distributed with some distribution $W[P]$ to be determined from the equation
\begin{eqnarray}
    &&\Psi_0(\gamma, C) =\Psi(\gamma, C)_{t=0} =\int [D P_0] W[\vec P_0]\exp{
	   \frac{-\I \gamma}{\nu}\oint \vec{C} \cdot d \vec{P}_0};
\end{eqnarray}
This path integral is nothing but a functional Fourier transform, which can be inverted as follows
\begin{eqnarray}\label{InvFourier}
    &&W[P_0] = \int [D C]  \Psi_0(\gamma, C) \exp{
	   \frac{\I \gamma}{\nu}\oint  \vec{C} \cdot d\vec{P}_0};
\end{eqnarray}
The definition of the parametric-invariant functional measure in this Fourier integral was discussed in detail in the old work \cite{M93, M23PR}. The periodic vector functions $\vec C(\theta), \vec P(\theta)$ are represented by the Fourier series, after which the measure becomes a limit of the multiple integrals over all the Fourier coefficients. 
As an alternative, one may replace these loops with polygons with $ N\to\8$ sides and define the measure as a product of integrals over the positions of the vertices of these polygons.

The explicit formulas for the Fourier measure, proof of its parametric invariance, and some computations of the Functional Fourier Transform can be found in \cite{M23PR}, section 7.10 (Initial data).

In the previous section, we completely solved the Cauchy problem for an interesting example -- the exact fixed point of the loop equation corresponding to a global random rotation. The resulting momentum loop was a singular periodic function of $\theta$ with a Gaussian distribution of parameters (vertices or Fourier coefficients).

In a physically justified case of Gaussian thermal noise $\vec \xi(\vec r)$ added to the initial velocity field $\vec v_0(\vec r)$, we can advance solving the Cauchy problem for a generic initial velocity field.

The potential component of $\vec \xi(\vec r)$, proportional to the gradient of some scalar, drops from the loop functional. Therefore we could always add such a term to make $\vec \nabla \cdot \vec \xi(\vec r) =0$, preserving incompressibility. Though such a constraint noise is not quite physical, it is equivalent to a physical Gaussian random noise inside the velocity circulation we study.

Averaging the initial loop functional over Gaussian noise, we find 
\begin{eqnarray}\label{NoisyPsi}
    &&\Psi_0(\gamma, C) = \exp{\frac{\I \gamma}{\nu}\oint_C d \vec r \cdot \vec v_0(\vec r)- \frac{\gamma^2}{2 \nu^2} \oint _C \oint _C  d \vec r \cdot \VEV{\vec \xi(\vec r)\otimes \vec \xi(\vec r') } \cdot d \vec r'}
\end{eqnarray}
This Gaussian noise is correlated at small distances $r_0$, related to the molecular structure of the fluid, which leads to the following estimate
\begin{eqnarray}
    &&\oint _C \oint _C  d \vec r \cdot \VEV{\vec \xi(\vec r)\otimes \vec \xi(\vec r') } \cdot d \vec r' \to \frac{|C| \sigma^2}{r_0^2} ;\\
    && |C| =\oint \abs{d \vec C(\theta)} = \int_0^1 d \theta | \vec C'(\theta)|
\end{eqnarray}
The steps leading to this estimate are as follows (assuming natural length parametrization with $\vec C'(s)^2 =1$):
\begin{eqnarray}
    &&\VEV{\xi_i(\vec r)  \xi_j(\vec r') } = \left(\delta_{i j} - \frac{\nabla_i \nabla_j}{\nabla^2}\right) g\left((\vec r-\vec r')^2\right);\\
    &&\oint_C d r_i\oint_C d'_j \left(\delta_{i j} - \frac{\nabla_i \nabla_j}{\nabla^2}\right) g\left((\vec r-\vec r')^2\right) =\nonumber\\
    &&\oint_C d r_i\oint_C dr'_i g\left((\vec r-\vec r')^2\right) =\oint d s \oint d s' \vec C'(s) \cdot \vec C'(s') g\left((\vec C(s) - \vec C(s'))^2\right)
\end{eqnarray}
Expanding $\vec C(s') \to \vec C(s) + (s'-s) \vec C'(s) , \vec C'(s') \to \vec C'(s)$ which is valid in small vicinity $s' \to s$ assuming the function $g(\vec \rho)$ is supported in a small vicinity $|\vec \rho| \sim r_0 \to 0$ we find
\begin{eqnarray}
    &&\oint_C d r_i\oint_C dr'_i g\left((\vec r-\vec r')^2\right) \to  m_0 |C|;\\
    && m_0 = \int_{-\8}^{\8} d x g(x^2) \propto \frac{\sigma^2}{r_0^2}
\end{eqnarray}
The last estimate follows from dimensional counting and normalization of the variance.
This estimate yields the following initial distribution of the random loop $\vec P_0(\theta)$
\begin{eqnarray}
    &&W[P_0] = \int [D C] \exp{- m_0 |C| +
	   \frac{\I \gamma}{\nu}\oint d \vec{C} \cdot \left(\vec v_0(\vec C(\theta)) - \vec{P}_0(\theta)\right)};\\
    && m_0 = \frac{\gamma^2 \sigma^2}{2 \nu^2  r_0^2}
\end{eqnarray}
This path integral is equivalent to that of a relativistic Klein-Gordon particle in the presence of the electromagnetic field with vector potential $\vec v_0(\vec r)$ in three Euclidean dimensions. The unusual feature is the distributed momentum $\vec P_0(\theta)$ along the loop.

Let us compute this path integral for the uniform initial velocity $\vec v_0(\vec r) = \const{}$.
In this case, the circulation is zero, so we are left with the Fourier transform of the exponential of the loop's length.

This path integral is equivalent \cite{PolyakovGFS} to the Klein-Gorgon propagator of the free massive particle with the mass $m_0$ up to renormalization coming from the short-range fluctuations of the path.

The constant velocity $v_0(\vec C(\theta))$ drops from the closed-loop integral, which brings the exponential to the ordinary momentum term in the Action $\oint d \vec C(\theta) \cdot \vec P_0(\theta)$.
This path integral is computed by fixing the gauge for the parametric invariance $\theta \Rightarrow f(\theta)$, which is studied in the modern QFT, say in \cite{PolyakovGFS}, Chapter 9.

For the mathematical reader, we reproduce this computation in the Appendix without any reference to the path integrals for the Klein-Gordon particle, using a polygonal approximation of the loop
The result is the following Gaussian distribution
\begin{eqnarray}
    &&W[P] \propto \int_0^\infty d T \exp{- \frac{\gamma^2}{2 \nu^2} \int_0^T d s\left(\frac{\sigma^2}{r_0^2} +\vec P(s)^2\right)};
\end{eqnarray}

Fourier coefficients $ p_\alpha(n)$ can parametrize this periodic trajectory
\begin{eqnarray}
    &&P_\alpha(s) = \sum_{n=-\infty; n \neq 0}^\infty  p_\alpha(n) \exp{\frac{2 \pi \I n s}{T}} ;\\
    &&  \bar p_\alpha(m) = p_\alpha(-m);\\
    && \VEV{p_\alpha(n)  p_\beta(m)}  = \frac{\delta_{\alpha\beta}\nu^2}{\gamma^2 T} \delta_{n,-m}
\end{eqnarray}
The term with $n=0$ is omitted, as it drops from the closed loop integral $ \int \vec C(s) \cdot d\vec P(s)$.
These Fourier coefficients at fixed $T$ are Gaussian variables with the above variance matrix $\VEV{p_\alpha(n)  p_\beta(m)}$.
This property is, in principle, sufficient to compute the terms of the perturbative expansions in inverse powers of viscosity (see below).

These Fourier coefficients do not decrease with number, so the curve $\vec P(\theta,0) $ is fractal rather than smooth. In particular, $\vec P(T) \neq \vec P(0)$.

Note that the limit $\sigma \to 0$ of smooth initial velocity field corresponds to the zero mass for this relativistic particle. This limit does not lead to infinities in correlation functions in three or more dimensions. 
This important question deserves more investigation in our case. \textbf{If this limit exists, we can prove the no-explosion theorem for the smooth initial field.}

Also note that the limit $\sigma \to 0, r_0 \to 0$ at fixed ratio 
$\frac{\sigma^2}{r_0^2}$ corresponds to local noise with vanishing variance. 
In this limit, velocity gradients become non-smooth, and thus the strict 
regularity conditions of the Millennium problem are not met. Nevertheless, 
even an infinitesimal local noise provides a physically meaningful regularization. 
Physically, a complete absence of fluctuations ($\sigma = 0$) corresponds to vanishing 
temperature in which case the fluid would freeze, rendering the Navier–Stokes 
equations physically irrelevant.

Let us summarize the results of this section. We bypassed the nonlinear Cauchy problem for the \NS{} equation by treating it as a limit of the solvable Cauchy problem in the linear loop equation.
As we argued, the unavoidable thermal noise in any physical fluid makes such a limit the correct definition.

We have advanced the Cauchy problem further by reducing the dimensionality from $d +1 $ dimensions in the \NS{} equation to $1 + 1$ dimensions in the MLE.
 
 Before elaborating on that dimensional reduction, we consider an exact solution of the loop equation corresponding to the random global rotation of the original velocity field and the associated Cauchy problem.

\section{Universality and Scaling of MLE}\label{Scaling}
Various symmetry properties affect solutions' space, especially their fixed trajectories.

First of all, this equation is parametric invariant:
\begin{eqnarray}
    \vec P(\theta, t) \Rightarrow \vec P(f(\theta), t);\; f'(\theta) >0;
\end{eqnarray}
Naturally, any initial condition $\vec P(\theta, 0) = \vec P_0(\theta)$ will break this invariance; each such initial data will generate a family of solutions corresponding to initial data $\vec P_0(f(\theta))$.

The lack of explicit time dependence on the right side leads to time translation invariance:
\begin{eqnarray}
    \vec P(\theta, t) \Rightarrow \vec P(\theta, t+ a)
\end{eqnarray}
Less trivial but also very significant is the time-rescaling symmetry:
\begin{eqnarray}
    \vec P(\theta, t) \Rightarrow \sqrt{\lambda}  \vec P(\theta, \lambda t),
\end{eqnarray}
This symmetry follows because the right side of \eqref{PloopEq} is a homogeneous functional of the third degree in $\vec P$ without explicit time dependence.

This scale transformation is quite different from the scale transformation in the \NS{} equation, which involves rescaling of the viscosity parameter:
\begin{eqnarray}
    && \vec v(\vec r, t) \Rightarrow \frac{\vec v(\alpha \vec r, \lambda t)}{\alpha\lambda} ;\\
    && \nu \Rightarrow \nu \frac{\alpha^2}{\lambda} 
\end{eqnarray}
In our case, there is a genuine scale invariance without parameter changes; in other words, no dimensional parameters are left in MLE.

Using this invariance, one can make the following transformation of the momentum loop and its variables
\begin{eqnarray}\label{PFT}
    \vec{P} = \sqrt{\frac{\nu}{2(t+ t_0)}} \frac{ \vec{F}}{\gamma}
\end{eqnarray}
The new vector function $\vec{F}$ satisfies the following dimensionless equation
\begin{eqnarray}\label{Fequation}
   &&2\partial_\tau \vec{F} = \left(1- (\Delta \vec{F})^2\right) \vec{F}  +\Delta \vec{F} \left(\gamma^2 \vec{F} \cdot \Delta \vec{F} +\I \gamma \left( \frac{(\vec{F} \cdot \Delta \vec{F})^2}{\Delta \vec{F}^2}- \vec{F}^2\right)\right);\\
   && \tau = \log \frac{t + t_0}{t_0};
\end{eqnarray}
The \textbf{viscosity disappeared from this equation}; now it only enters the initial data
\begin{eqnarray}
     \vec F(\theta, 0) = \sqrt{\frac{2 t_0}{\nu}} \vec P_0(\theta)
\end{eqnarray}
This universality property is extremely important.
Note that the loop functional is now represented as
\begin{eqnarray}
\label{PsiFsol}
    &&\Psi(C,t) = \VEV{\exp{\frac{\displaystyle \imath \oint  d\vec C(\theta) \cdot \vec F\left(\theta, \log \frac{t + t_0}{t_0}\right)}{\sqrt{2 \nu (t + t_0)}}}} 
\end{eqnarray}
with the square root of viscosity in the denominator as a coupling constant in nonlinear QFT. The averaging  $\VEV{\dots}$ goes over the manifold of solutions $\vec F(\theta,\tau)$ of the ODE \eqref{Fequation}.

This formula immediately suggests that turbulence is a quasiclassical phenomenon in our quantum mechanical system that can be studied by the well-known WKB method (or corresponding methods developed by Kolmogorov and Maslov in the mathematical literature).

In the conventional approach to fluid mechanics, based on the \NS{} equation, the Reynolds number distinguishing between the laminar and turbulent flow enters the equation and controls the relative magnitude of nonlinearity. One has to study various regimes in that equation, including the inviscid limit presumably related to the turbulence, but different from the Euler equation due to the dissipation anomaly.

In our dual theory, representing the same \NS{} dynamics as a quantum system, the dynamical equation \eqref{Fequation} is universal; it does not depend upon the Reynolds number.
This number enters initial data and the relation between our solution for $\vec F$ and the loop functional (i.e., the PDF for the circulation as a functional of the shape of the loop).

The evolution of the loop functional $\Psi$ inside the unit circle in the complex plane in Fig.\ref{fig::LoopFunctionalEvolution} goes by universal trajectories, determined by \eqref{Fequation}. The Reynolds number describes the initial position of this $\Psi$ inside the circle. The distance $|\Psi-1|$ from the fixed point $\Psi_* =1$ is the true measure of turbulence. One could expect a laminar flow solution in some small vicinity of this fixed point (corresponding to potential flow).
\section{Vorticity Statistics and Momentum Loop}

The multiple area derivatives of $\Psi$ yield the formula  
\begin{eqnarray}
   && \VEV{\exp{\left(\imath \Gamma_C[v]\right)} \omega(r_1)\otimes \omega(r_2) \dots \otimes\omega(r_n)}_v = \nonumber\\
    &&\frac{n!}{(2 \pi)^n}\VEV{\exp{\imath \Gamma_{\tilde C}[P]} \idotsint\limits_{0<\theta_1<\dots<\theta_n<2\pi} d\theta_1 \hat \omega(\theta_1)\otimes d\theta_2 \hat \omega(\theta_2)\dots \otimes d\theta_n\hat \omega(\theta_n)}_P;\\
    \label{omegadef}
    && \hat \omega(\theta) \propto \imath\vec P(\theta)\times \Delta\vec P(\theta);\\
    &&\tilde C = \oplus_k L(r_C, r_k, r_C);
\end{eqnarray}  
where $L(r_C, r_k, r_C)$ is a double line from the center of mass, $r_C = \sum_{k=1}^n r_k/n$, to each point $r_k$ and back to $r_C$. The entire loop $C$ resembles bicycle spokes without a wheel (see Fig.~\ref{fig::BykeWheel}). The angles $\theta_1, \dots, \theta_n$ are ordered on the unit circle.  

\begin{figure}[h]
    \centering
    \includegraphics[width=0.5\textwidth]{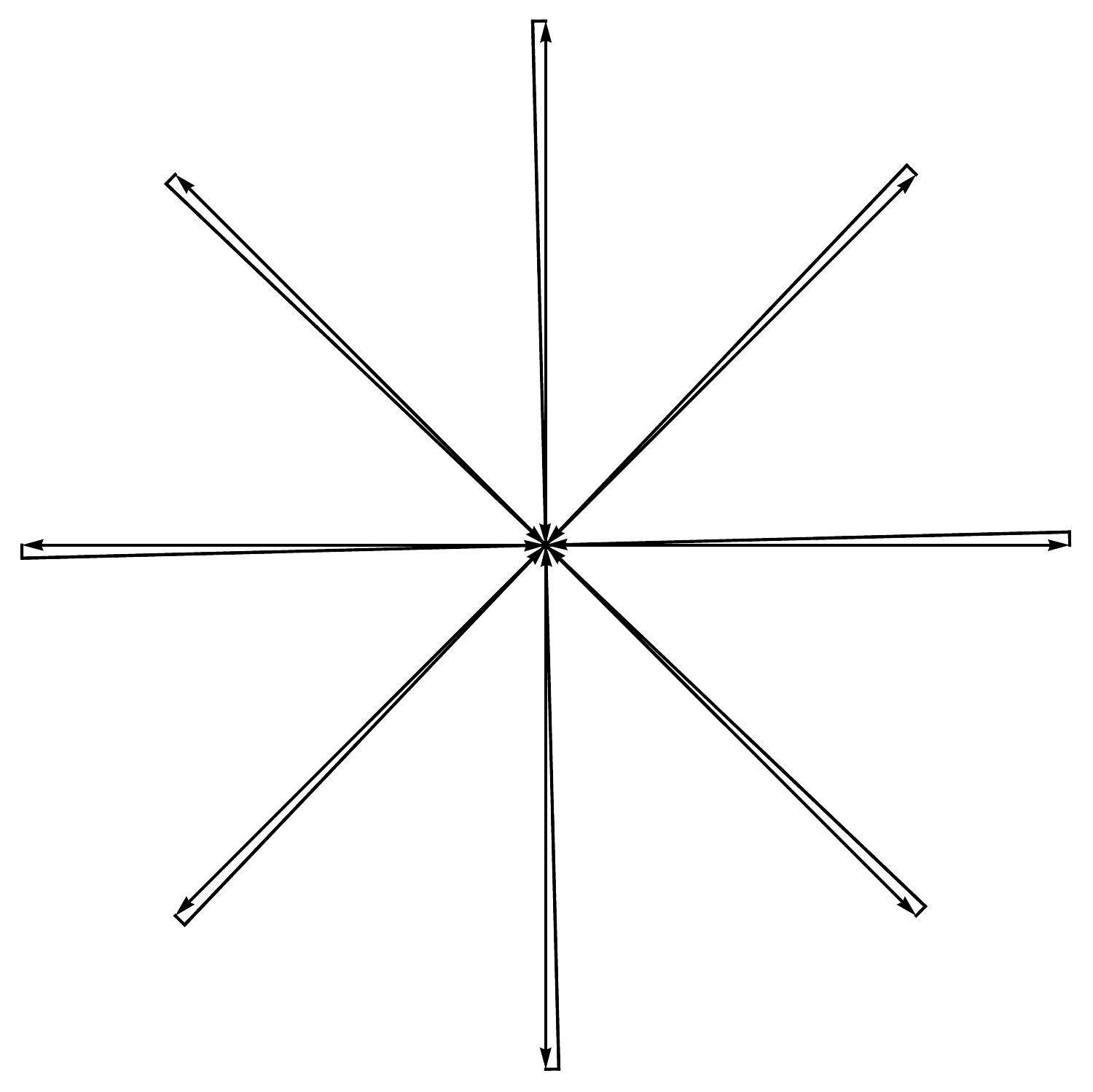}
    \caption{The loop $C$ used to compute vorticity correlation functions. The endpoints of each spoke correspond to $\vec C(\theta_k) = \vec r_k$, while the center represents the center of mass, $\vec r_C = \frac{\sum_1^n \vec r_k}{n}$. Each spoke consists of two segments: one from the center to $\vec r_k$, and another returning to the center. Since the area enclosed by the loop is zero, the circulation vanishes, i.e., $\Gamma_C[v]=0$.}
    \label{fig::BykeWheel}
\end{figure}

These double lines cancel in the circulation,  
\begin{equation}
    \Gamma_C[v]=0,
\end{equation}
but not in the momentum loop representation. This occurs because, in the momentum loop integral  
\begin{equation}
    \int d\theta \vec C'(\theta)\cdot \vec P(\theta),
\end{equation}
the momentum variable $\vec P(\theta)$ depends explicitly on $\theta$, preventing the cancellation of the two line integrals in a double line.  

The factor $\frac{n!}{(2 \pi)^n}$ arises as follows. In the velocity representation, where the circulation is zero, the vorticities  
\begin{equation}
    \vec\omega(\vec C(\theta_k)) = \vec \omega(\vec r_k)
\end{equation}
depend directly on $\vec r_k$, rather than on the angles $\theta_k$. Consequently, the integral over ordered angles $\theta_k$ yields the factor $\frac{(2 \pi)^n}{n!}$, which we must compensate with the inverse factor.  

This leads to the fundamental formula for multiple vorticity correlation functions:  
\begin{eqnarray}\label{MomLoopCorr}
   && \VEV{\omega(r_1)\otimes \omega(r_2) \dots \otimes\omega(r_n)}_v = \nonumber\\
    && \frac{n!}{(2 \pi)^n}\VEV{\exp{\imath \Gamma_{\tilde C}[P]} \idotsint\limits_{0<\theta_1<\dots<\theta_n<2\pi} d\theta_1 \hat \omega(\theta_1)\otimes d\theta_2 \hat \omega(\theta_2)\dots \otimes d\theta_n\hat \omega(\theta_n)}_P;\\
    && \Gamma_{\tilde C}[P] = \sum_{k=1}^n \int_0^1 d \eta (\vec r_k - \vec r_C) \cdot \vec Q_k(\eta);\\
    && \vec Q_k(\eta) = \vec P\left(\tilde\theta_k(1-\eta) + \theta_k \eta\right) - \vec P\left(\tilde\theta_{k+1}\eta + \theta_k (1-\eta)\right);\\
    && \tilde\theta_k = \frac{\theta_{k-1} + \theta_k}{2} \mod 2 \pi.
\end{eqnarray}  

This formula expresses the multiple correlation functions in terms of expectation values within the Euler ensemble. In \cite{migdal2024quantum}, we use this approach to compute the two-point correlation function, $\VEV{\omega \omega}$, as a function of relative distance for $n=2$. The Fourier-space correlation function,  
\begin{equation}
    \VEV{\vec \omega(k)\cdot\vec \omega(-k)}
\end{equation}
(i.e., the energy spectrum), is found to be real and positive.  

We conclude that the full statistical description of the rotational part of the velocity field (i.e., vorticity) is related to the statistics of the solutions $\vec P(\theta, t)$ of the MLE.  

A crucial aspect of this relation is the violation of time-reversal symmetry. Under time reversal, vorticity changes sign, which would typically eliminate vorticity correlations for an odd number of points ($n$). However, this is not the case in the momentum loop representation \eqref{MomLoopCorr}.  

To see why, consider the complex conjugation of the loop functional, which corresponds to time reversal in the original velocity representation. The vorticity operator $\hat \omega(\theta)$ in \eqref{omegadef} contains an imaginary unit, so it changes sign. The $P$-circulation term in the exponential is real but changes sign under the coordinate reversal $\vec r_k \Rightarrow - \vec r_k$, thereby compensating for the complex conjugation effect.  

This results in the identity  
\begin{eqnarray}
     && \VEV{\omega(r_1)\otimes \omega(r_2) \dots \otimes\omega(r_n)}_v = (-1)^n \VEV{\omega(-r_1)\otimes \omega(-r_2) \dots \otimes\omega(-r_n)}_v,
\end{eqnarray}  
which reflects parity conservation but does not enforce time-reversal symmetry. As a consequence, the expectation value of an odd number of vorticities remains an odd function of the coordinates.  

More precisely, in both even and odd cases, the expectation values remain real, but they are linked to the real and imaginary parts of the exponential, depending on the number of imaginary units ($\imath$) in the product of vorticities:  
\begin{eqnarray}
    && \VEV{\omega(r_1)\otimes \omega(r_2) \dots \otimes\omega(r_n)}_v = \nonumber\\
    &&\frac{n!}{(2 \pi)^n}
   \Re \VEV{\exp{\imath \Gamma_{\tilde C}[P]} \idotsint\limits_{0<\theta_1<\dots<\theta_n<2\pi} d\theta_1 \hat \omega(\theta_1)\otimes d\theta_2 \hat \omega(\theta_2)\dots \otimes d\theta_n\hat \omega(\theta_n)}_P.
\end{eqnarray}  

\section{Laminar flow at small time and seeds of turbulence}
The viscosity enters the MLE's denominator, making it straightforward to investigate the laminar flow (large viscosity) and even turbulent flow (small viscosity).

Let us start with the laminar flow. It corresponds to small $\vec F$, in which case the equation \eqref{Fequation} linearizes and can be explicitly solved
\begin{eqnarray}
    \vec F(\theta, t) \to \vec P_0(\theta) \sqrt{\frac{2(t_0 + t)}{\nu}} + O(F^3)
\end{eqnarray}
This solution will stay smooth when starting with the smooth initial value $\vec P_0(\theta)$. There will be no discontinuity in $\vec F(\theta,t)$ and no discontinuity in $\vec P(\theta,t)$.

For the loop functional this means zero area derivative, in other words, potential flow without vorticity.
Moreover, this flow will stay as a potential flow in every order of the formal perturbation expansion in inverse powers of $\nu$ for an arbitrary smooth initial value $\vec P_0(\theta)$.

However, any finite initial discontinuity in $\vec P_0(\theta)$ would lead to nontrivial terms of this perturbation expansion. These terms will be singular but scale as higher powers of $\Delta \vec F$. One may expect these corrections to be controlled at a large enough viscosity (compared to initial circulation).

The above thermal fluctuations lead to a small but singular contribution to the initial momentum loop. The Fourier coefficients $ \vec p_n$ do not decrease with order $n$, leading to the delta function singularity in the correlation function $\VEV{\vec P(\theta)\otimes \vec P(\theta') }\propto \delta(\theta-\theta')$, which is stronger than the discontinuity, required for the presence of vorticity.

After sufficient time, these small singular terms may lead to larger singular terms in the solution.

The recent paper \cite{Bandak_2024} argued that the thermal fluctuations could produce turbulence in finite time, comparable with experimental times of the large eddy formation. In other words, these small fluctuations could quickly grow and end up as large random eddies observable in experiments by order of magnitude estimates in \cite{Bandak_2024}.

 Our theory considers two possible asymptotic regimes: decaying turbulence or a finite-time explosion.
We study these regimes in the subsequent sections.

\section{Decaying turbulence}\label{DeTur}
The solutions originating deep inside the unit circle, far from $\Psi=1$, can become turbulent and eventually decay to $\Psi \to 1$ due to energy dissipation by vorticity structures. This decay for $\vec P(\theta,t)$ corresponds to the fixed point equation for $\vec F$
\begin{eqnarray}
\label{FP}
     && \left((\Delta \vec{F})^2 -1\right)\vec F   =\Delta \vec{F} \left(\gamma^2 \vec{F} \cdot \Delta \vec{F} +\I \gamma \left( \frac{(\vec{F} \cdot \Delta \vec{F})^2}{\Delta \vec{F}^2}- \vec{F}^2\right)\right)
\end{eqnarray}
This fixed point $\vec F(\theta)$ is not a solution of the Cauchy problem for the loop functional, though we expect the solution of some Cauchy problems to asymptotically approach this fixed point at a large time.

This fixed point represents the solution of the loop equation \textbf{with the boundary condition $\Psi(\theta,+\8) =1$}. This boundary condition describes the flow eventually stopping as a result of the dissipation of kinetic energy $$E =\int d^3 r \frac{\vec v^2}{2}, \;  \partial_t E = - \nu \int d^3r \vec \omega^2 <0.$$
\subsection{Fixed point solution}
The saddle point equation \eqref{FP} was solved and investigated in previous papers \cite{migdal2023exact, migdal2024quantum}. 
The solution for $\vec F(\theta)$ is a fractal curve defined as a limit $N \to \infty$ of the polygon $\vec{F}_0\dots \vec{F}_N= \vec{F}_0$ with the following vertices
\begin{eqnarray}\label{Fsol}
    && \vec{F}_k =\Omega \cdot \frac{\left\{\cos (\alpha_k), \sin (\alpha_k), \I\cos \left(\frac{\beta }{2}\right)\right\}}{2 \sin \left(\frac{\beta }{2}\right)} ;\\
    && \theta_k = \frac{k}{N}; \; \beta = \frac{2 \pi p}{q};\; N \to \infty;\\
    &&\alpha_{k} = \alpha_{k-1} + \sigma_k \beta;\; \sigma_k =\pm 1,\;\beta \sum\sigma_k = 2 \pi p r;\\
    && \Omega \in SO(3)
\end{eqnarray}
The parameters $ \hat{\Omega},p,q,r,\sigma_0\dots\sigma_{N}= \sigma_0$ are random, making this solution for $\vec{F}(\theta)$ a fixed manifold rather than a fixed point.
We suggested in \cite{migdal2023exact} calling this manifold the big Euler ensemble of just the Euler ensemble.

It is a fixed point of \eqref{FP} with the discrete version of discontinuity and principal value:
\begin{eqnarray}
    &&\Delta \vec{F} \equiv \vec{F}_{k} - \vec{F}_{k-1};\\
    && \vec{F} \equiv \frac{ \vec{F}_{k} + \vec{F}_{k-1}}{2}
\end{eqnarray}
Both terms of the right side \eqref{Fequation} vanish; the coefficient in front of $\Delta \vec{F}$ and the one in front of $\vec{F}$ are both equal zero.
Otherwise, we would have $\vec{F} \parallel \Delta \vec{F}$, leading to zero vorticity \cite{migdal2023exact}.

This requirement leads to two scalar equations
\begin{subequations}\label{twoeqs}
  \begin{eqnarray}
    && (\Delta \vec{F})^2 =1;\\
    && \vec{F}^2 - \frac{\gamma^2}{4} =  \left(\vec{F} \cdot \Delta \vec{F}- \frac{\I \gamma}{2}\right)^2;
\end{eqnarray}  
\end{subequations}

The integer numbers $\sigma_k= \pm1$ came as the solution of the loop equation, and the requirement of the rational $\frac{p}{q}$ came from the periodicity requirement, as we prove below.

In our limit, the integral for velocity circulation becomes the Lebesque sum:
\begin{eqnarray}
\label{CircByParts}
    && \oint d \vec{C}(\theta) \cdot \vec{F}(\theta) \to
    \sum_k\Delta \vec{C}_k \cdot \vec{F}_k;
\end{eqnarray}

A remarkable property of this solution $\vec{F}(\theta,t)$ of the loop equation is that even though it satisfies the complex equation and has an imaginary part, the resulting circulation \eqref{CircByParts} is real!
The imaginary part of the $\vec{F}_k$ does not depend on $k$ and thus drops from the total sum $\sum_k\Delta \vec{C}_k =0$ due to the periodicity of the loop $C$.

Another noteworthy observation is that the solution \eqref{Fsol} exhibits a symmetric distribution: $-\vec{F}_k$ and $\vec{F}_k$ share the same PDF due to integration over all possible rotation matrices. A rotation by $\pi$ in the $xy$ plane, combined with complex conjugation, leaves the distribution invariant. Moreover, as we have seen, the imaginary part of $\vec{F}_k$ does not contribute to the loop functional. Consequently, the PDF of $ \Gamma = \sum_k \Delta \vec{C}_k \cdot \vec{P}_k$ is an even function.

However, multiple vorticity correlations are determined via the area derivative. The corresponding vector $\hat{\omega}_k$ in \eqref{omegaK} is quadratic in $\vec{P}$, and as a result, this reflection symmetry does not suppress the expectation value of an odd number of $\omega_k$ factors. For instance, the triple correlator $\oal(1) \obe(2) \oga(3)$ remains nonzero. The corresponding triple velocity correlator, $ \val(1) \vbe(2) \vga(3)$, can be obtained via Fourier transformation, where $\vec{v}_k = \I \vec{k} \times \vec{\omega}_k / \vec{k}^2$, up to purely potential terms linear in the coordinates.

These potential terms, however, do not contribute to the energy flow in wavevector space, contrary to popular belief, as discussed in \cite{migdal2024quantum}. Instead, they generate gradients of the delta function, $\pd{\vec{k}} \delta^3(\vec{k})$, rather than a constant energy flux. Such terms are influenced by boundary conditions at infinity and, therefore, do not represent spontaneous stochasticity caused by random vortex structures within the bulk.

\subsection{The proof of the Euler ensemble as a fixed point of MLE}
Let us present here the proof of this solution, verified by \Mathematica{} (see \cite{DecayTurb23}).
\begin{mytheorem}
The Euler ensemble solves the discrete MLE.
\end{mytheorem}
\begin{proof}

We start from the general Anzatz with real vectors $\vec A, \vec f_k$ , corresponding to the real circulation in \eqref{CircByParts}
\begin{eqnarray}
    &&\vec F_k = \I \vec A + (\vec f_{k -1} + \vec f_k)/2;\\
    &&\Delta \vec F_k = \vec f_{k} - \vec f_{k-1};\\
    && (\vec f_{k} - \vec f_{k-1})^2=1
\end{eqnarray}
Analyzing the imaginary and parts of the second equation in \eqref{twoeqs}, we observe that the imaginary part will vanish provided
\begin{eqnarray}
    &&\vec A \cdot \vec f_k =0 \forall{k};\\
    && \vec f_k^2 = \vec f_{k-1}^2 \forall{k};
\end{eqnarray}
We conclude that $\vec f_k$ belongs to a circle with some radius $R$ in the origin of the plane,  which plane is orthogonal to $\vec A$.
In the coordinate frame where $\vec A = \{0,0,A\}$ 
\begin{eqnarray}
    \vec f_k = R \left\{\cos(\alpha_k), \sin(\alpha_k), 0\right\}
\end{eqnarray}
The $SO(3)$ matrix needed to rotate our vectors to this coordinate frame can be absorbed into the rotation matrix $\Omega$ we have in our solution.

The radius $R$ and $A$ are determined by the real part of our equations as follows
\begin{subequations}
    \begin{eqnarray}
  &&4 A^2 = 2 R^2 \left(1 + \cos(\alpha_{k} - \alpha_{k-1})\right);\\
  && 1 = 2 R^2 \left(1 - \cos(\alpha_{k} - \alpha_{k-1})\right);
\end{eqnarray}
\end{subequations}
Solving these two equations, we find the $\mathbb{Z}_2$ variables at every step
\begin{eqnarray}
    \alpha_{k} = \alpha_{k-1} + \beta \sigma_k,\; \sigma_k^2 =1;
\end{eqnarray}
The radius $R$ and the length $A= |\vec A|$ are related to this angular step $\beta$
\begin{eqnarray}
   &&R =  \frac{1}{2 \sin \left(\frac{\beta }{2}\right)};\\
   && A = \frac{1}{2 \tan \left(\frac{\beta }{2}\right)};
\end{eqnarray}
The periodicity of the sequence $\vec f_k$ requires the angular step to be a fraction of $2 \pi$, which brings us to the Euler ensemble \eqref{Fsol}.
\end{proof}
\pctPDF{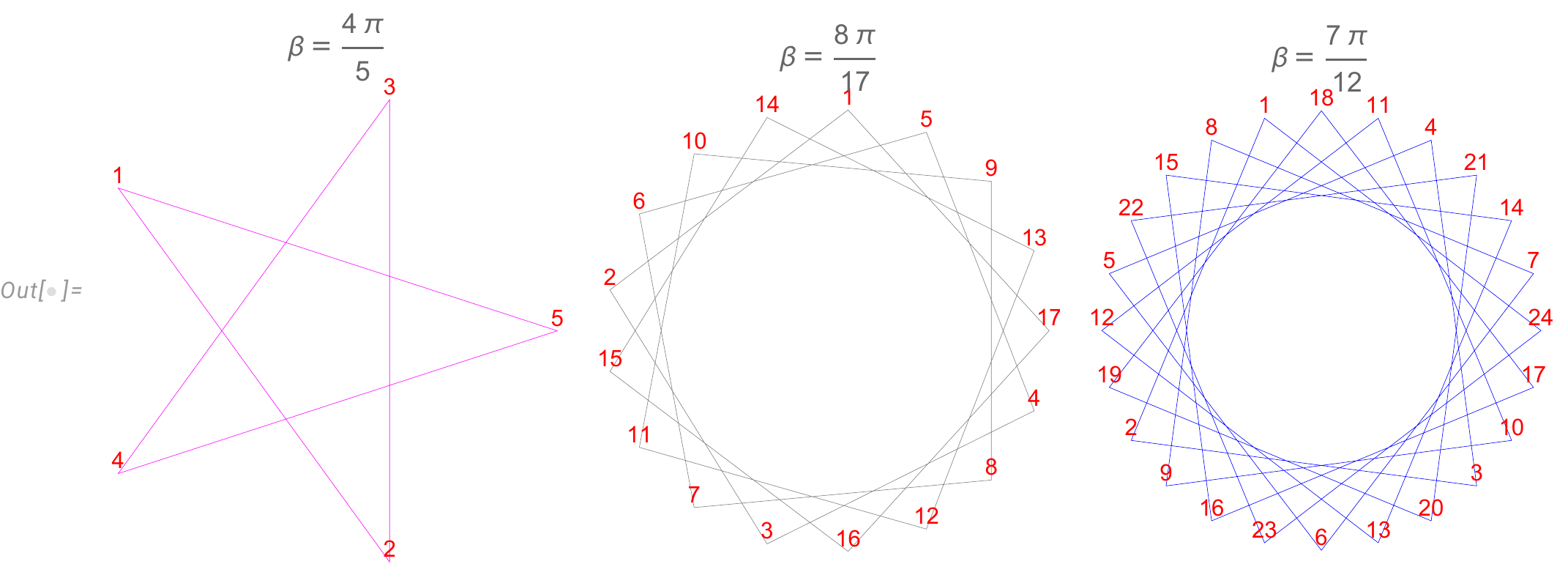}{regular star polygons for Euler ensembles of various $p,q$. The $\sigma_k$ variable indicates the direction of the random step of the link  $k\leftrightarrow k+1$. The random walk could go several times around the polygon as long as it ends where it started.}
\subsection{Euler ensemble as a random walk on a regular star polygon}
Geometrically, the vertices $\vec f_k$ belong to the regular star polygon with $q$ sides of unit length, with vertices at $R \exp{\I k \beta}, k= 1\dots q$. They were classified by Thomas Bradwardine (archbishop of Canterbury) and later by Johannes Kepler in the 17th Century and are denoted as $\{q/p\}$ ( so-called Schläfli symbol). 

We show several examples in Fig. \ref{fig::RegularPolygons}. The general polygon is characterized by co-prime $p, q$ with $ p < q < N, (N-q) = 0 \mod 2 $.
Euler totients count these polygons. The number $ N > q$ counts the coordinates $\vec f_k$  covering our polygon several times, so that, in general, each geometric vertex is covered more than once. 

The Ising variables $\sigma_k$ describe a random walk around this polygon with the extra condition that it comes to the initial point after $N$ steps. 
The random walk goes $k \leftrightarrow k+1$
according to the sign of $\sigma_k$. 
The periodicity condition requires $\beta $ to be a rational fraction of $2 \pi$. 

This quantization of the angle and the radius brings the number theory to the statistical distribution. Each polygon may be covered several times during this random walk with this periodicity condition. A certain winding number $w$ is related to $ \sum_1^N\sigma_i = q r, w = p r$.
Surprisingly, such a fundamental random walk problem on the 500-year-old geometric manifold has been solved only now. 

\subsection{Euler ensemble as string theory with discrete target space}
This random walk problem can also be interpreted as a closed fermionic string in the discrete target space consisting of regular star polygons on a (randomly rotated) plane. Integrating over fermionic degrees of freedom in a quantum trace of the evolution operator is equivalent to summation over occupation numbers $n_k = {0,1}$, providing directions $\sigma_k = 2 n_k-1$ of the random walk. 

The target space coordinates are the vertices of the regular star polygons  $\{q,p\}$. 

The integration over target space made of these regular star polygons becomes a discrete sum over states of the Euler ensemble: the fraction $\frac{p}{q}$,  the configurations of fermionic occupation numbers $\nu_k = 0,1$ and the winding number $w = \frac{p}{q}\sum ( 2 \nu_k-1) $.

Placing these polygons for a fixed $N$ on a torus in 3D space ordered by the angle $\beta$ shows the world sheet of our discrete string in Fig.\ref{fig::StringView}, with red/green colors of sides indicating random directions of random walk (occupation number of fermions). The large disk (infinite at $N=\8$) corresponds to endpoints $\beta = \frac{2 \pi}{N},  \frac{2 \pi(N-1)}{N}$.

\pctWPDF{1}{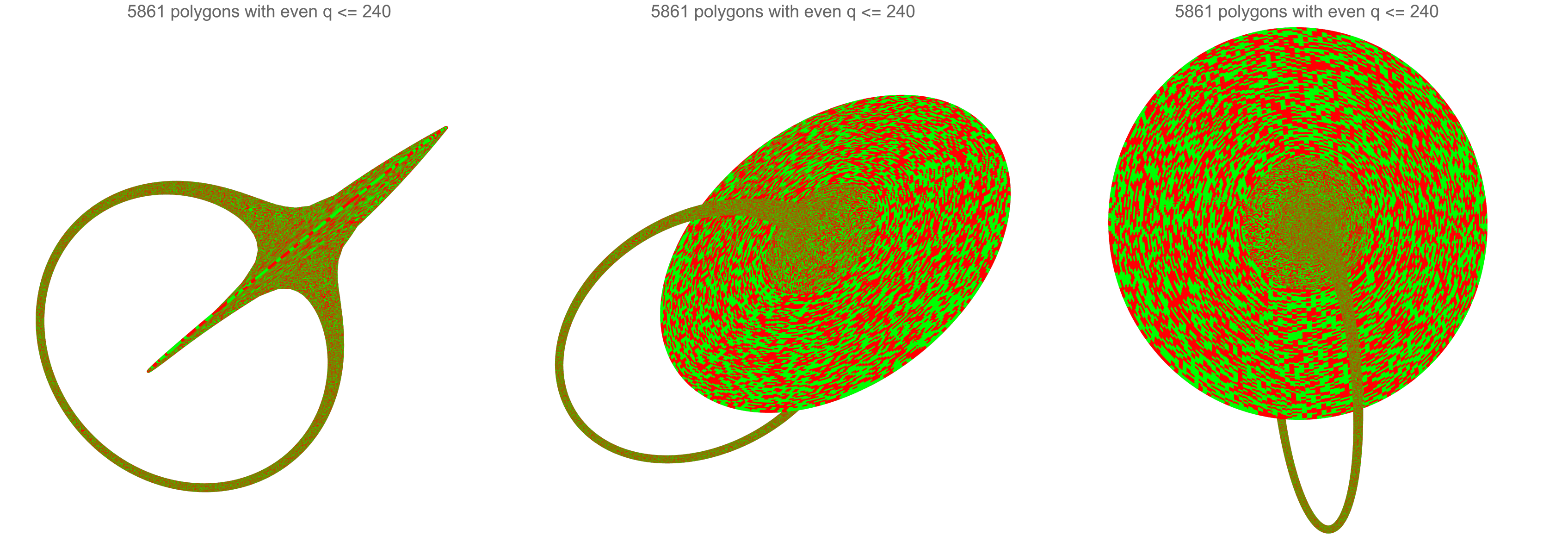}{The world sheet of our discrete string made of regular star polygons with unit side. The red/green colors of the sides indicate random directions of random walk.}
The solution of the Euler ensemble\cite{migdal2023exact} is based on new number theory identities for sums of powers of cotangent of fractions of $\pi$. These identities relate these sums to Jordan multi-totient functions weighted with Bernoulli coefficients.

The nontrivial part of using the Euler ensemble is the formula \eqref{PsiFsol} relating this ensemble to the observable loop functional of the decaying turbulence theory.

In the string theory language, where the momentum loop is the target space along with fermionic occupation numbers, this formula is the dual amplitude for the discrete string theory, with $\frac{\Delta \vec C(\theta)}{\sqrt{2\nu (t + t_0)}}$  playing the role of external momentum distributed along the closed string position (regular star polygon) $\vec F(\theta)$.

\textbf{This turbulence/string duality reveals the hidden beauty of primes under the ugly mask of chaos in the observable turbulent flow.}
 
 The corresponding universal energy spectrum for the decaying turbulence was computed in quadrature \cite{migdal2024quantum} in the quasiclassical limit at $\nu = \tilde \nu/N^2 \to 0$, and it closely matched the data of real and numerical experiments.
 The detailed comparison with available real and numerical experiments in decaying turbulence was published in the previous work \cite{QuantumWebsite, migdal2024quantum}. 
 Let us show here the figures from this work Figure.\ref{fig::EDecayData}, \ref{fig::NewDNSFit} demonstrating the match between our theory and the experiments (real and numerical).
 \pctPDFRot{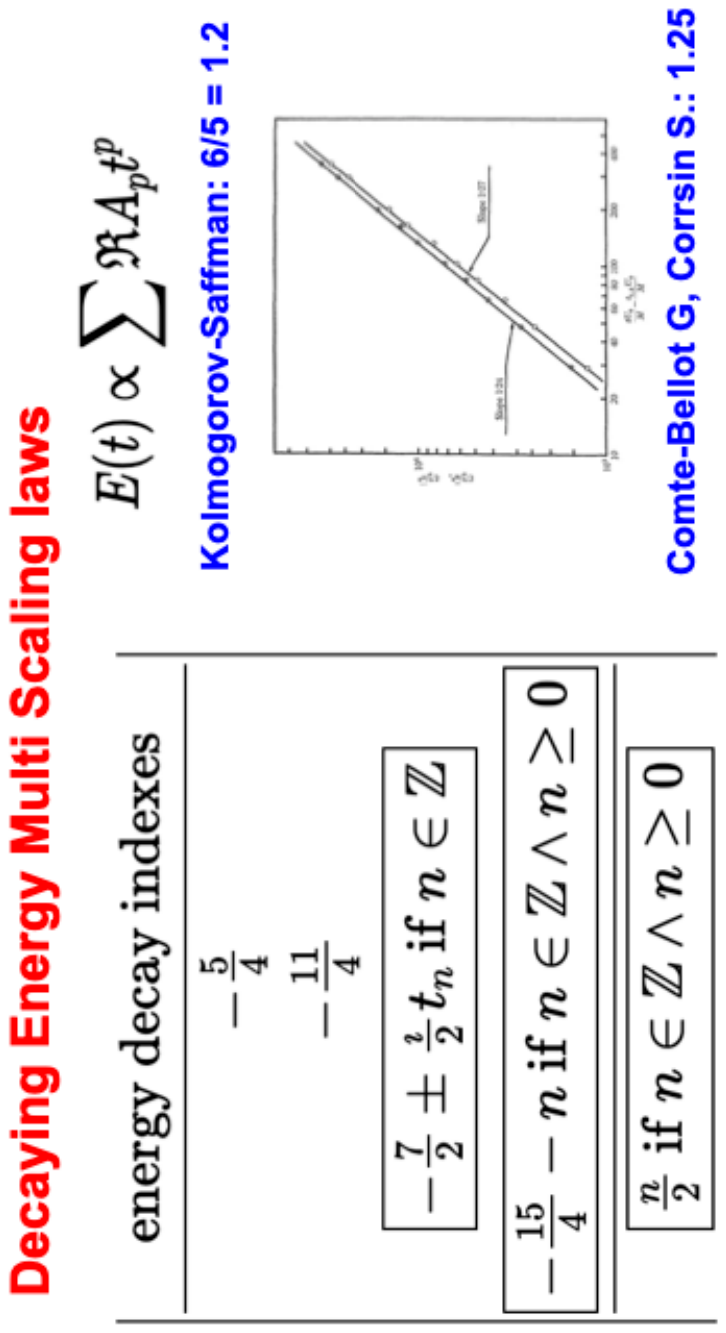}{The predictions of this theory compared with grid turbulence decay data from 1966. This data significantly deviate from Kolmogorov-Saffman model prediction $1.20$ but perfectly matches the leading decay index $\frac{5}{4}=1.25$ predicted by our theory.}
 \pctPDFRot{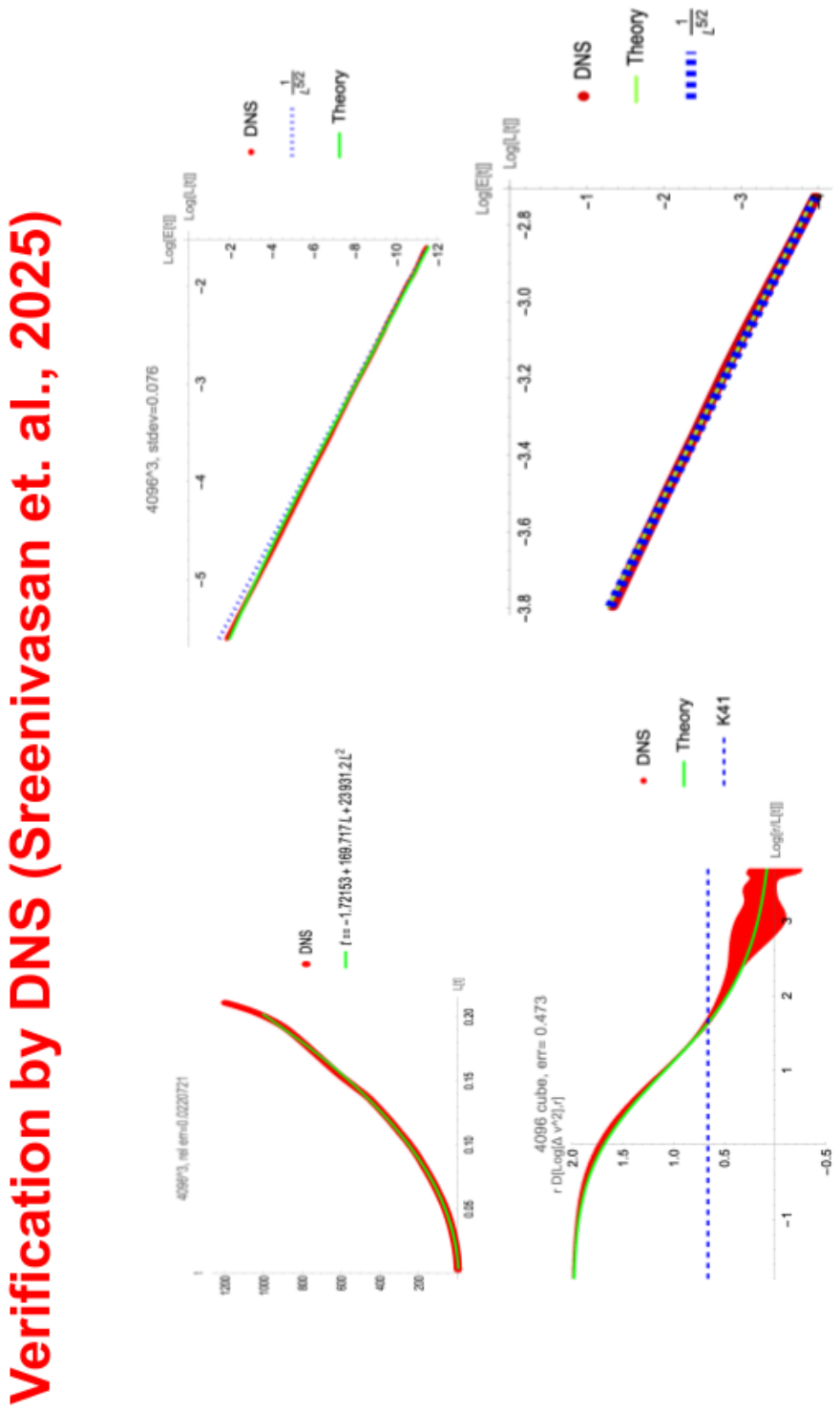}{We used the raw data from the 2025 DNS by A. Rodhiya and K.R. Sreenivasan (lattice $4096^3$) and the older one, (2024) by J.J Panicacheril, D. Donzis, and  K.R.Sreenivassan (lattice $1024^3$). 
 In the lower-left corner, the effective index for the second moment of velocity is plotted as a function of $\log r/\sqrt{t+ t_0} $ in the turbulent range.
Both data sets perfectly match the theoretical curve (green line) within the error bars. We only display the larger set results, which extend to lower values of the index, below 0.5. The errors are small in the upper part of the curve(above 0.5) and rapidly grow below that value, as it corresponds to large coordinate scales, where the lattice artifacts dominate.
Both matching curves theory/DNS deviate very far from the prediction of the K41 model $\eta = \frac{2}{3}$ .
On the upper left, the effective length $L(t)$ is compared with our prediction $\sqrt{t+ t_0}$ in the form of inverse function $t = -t_0 + a L^2$. It matches well in a wide time interval, which we interpret as decaying turbulence range. On the right side, there are two plots comparing decaying energy with the predictions (green) and the simple power law $ E \propto L(t)^{-\frac{5}{2}}$ for the two DNS data. There is a perfect match in the turbulent time range. Note that the dotted line ($ E \propto L(t)^{-5/2}$) deviates from the DNS data (red dots). This deviation reflects the subleading power terms included in the theoretical curve (green), which perfectly matches the DNS data within error bars.}
 
\subsection{The limit of large Reynolds number is not equivalent to vanishing viscosity}

The limit $\nu \to 0$ is tricky because $\nu$ is a measurable physical parameter with dimension $L^2/T$, and it fixes the magnitude of observable quantities. The real limit is the large Reynolds number $\RE = \frac{|\Gamma|}{\nu} \to \8$ where $\Gamma $ is a scale of circulation in the problem.
Our formula for the dissipation is 
\begin{eqnarray}
    &&\mathcal E = \nu \VEV{\vec \omega^2} \propto\frac{1}{\nu t^2} \VEV{ \sum_{k,n}\left(\vec F_k \times \Delta \vec F_k\right)\cdot \left(\vec F_n \times \Delta \vec F_n\right)  \exp{ \frac{\I\Gamma_F}{\sqrt{2\nu t}}}}_F;\\
    && \Gamma_F = \sum_k \Delta \vec C_k \cdot\vec F_k;
\end{eqnarray}
Mathematically, the same turbulent limit $|\Gamma_F| \gg \sqrt{\nu t}$ can be achieved by tending $\nu \to 0$, as only the ratio enters the exponential.

This limit will exist in our theory if we balance the powers of $\nu \to 0$ with powers of $N \to \8$ to keep the dissipation finite. This requires $\nu \sim 1/N^2$ as we found in the first paper \cite{migdal2023exact}.

The dimensionless parameters and functions of scaling arguments such as $|\vec r|/L(t), |\vec k| L(t)$ where $ L(t) = \sqrt{\tilde \nu (t + t_0)} $ all stay finite in this limit.

Other observable quantities, such as the energy spectrum $ E(k,t) = \oh\VEV{|\vec v_{\vec k}|^2 |}|\vec k|^2 $ or kinetic energy $E(t) = \int E(k,t) d k$ are not initially proportional to $\nu$ which means that the powers of $N $ will not balance in these quantities.

The resolution of this paradox is as follows \cite{migdal2024quantum}. The two-point vorticity correlation function must have a pole in the mathematical limit $\nu \to 0$ to provide finite energy dissipation:
\begin{eqnarray}\label{nuResidue}
    \VEV{\vec \omega^2} = \frac{F(\RE)}{\nu}
\end{eqnarray}
The residue $F( \8)$ in this pole is what we call the turbulent limit 
\begin{eqnarray}
    F(\8) = \lim_{N\to \8, \nu \to 0} \nu \VEV{\vec \omega^2}
\end{eqnarray}

We use the mathematical limit $N\to \8, \nu \to 0$ to compute this residue, but after that, we come back to the physical formula \eqref{nuResidue} with observable viscosity $\nu$ and $\RE=\8$.

So, we used the fictitious limit of zero viscosity to compute the residue of the observable correlator at zero viscosity.

\subsection{Continuum limit exists for the loop functional but not for the momentum loop}

There is a following peculiarity in our theory. The momentum loop $\vec P(\theta,t)$ has no continuum limit when $N \to \8$, but the original loop functional $\Psi(C,\gamma,t)$ stays finite in such a limit. To be more precise, there is renormalizability: the energy dissipation stays finite when $\nu \to \frac{\tilde \nu}{N^2} \to 0$.

The second moment of velocity difference $$\VEV{\Delta v^2} =\VEV{(\vec v(\vec r)-\vec v(0))^2}$$ as a function of scaling variable $|\vec r|/L(t)$ decays at small argument as an infinite series of power terms with growing positive powers $\frac{1}{t}\left(|\vec r|/L(t)\right)^{p}$.
The spectrum of these scaling indexes $p$ (unrelated to a dilatation operator as far as we know)  is given in the following table:
\begin{eqnarray}
\label{VVSpectrum}
    \left|
\begin{array}{c}
\text{scaling indexes $p$ of $\VEV{\Delta v^2}$} \\
\hline
  \fbox{$2 n\text{ if }n\in \mathbb{Z}\land n\geq 1$}\\
 \frac{5}{2}\\
 \frac{11}{2}\\
 \fbox{$7 \pm\imath t_{n}\text{ if }n\in \mathbb{Z}$} \\
 \fbox{$\frac{1}{2} (15+4 n)\text{ if }n\in \mathbb{Z}\land n\geq 0$} \\
\end{array}
\right|
\end{eqnarray}
Here $\pm t_n$ are imaginary parts of the zeros of $\zeta(z)$, all located at the ray $\Re z = \oh$, according to the Riemann hypothesis. At least, it was proven that these zeros are all inside the strip $ 0 < \Re z < 1$.

The second moment does not scale as a single power. The effective index (log derivative of the second moment as a function of coordinate difference), is a nontrivial function of the scaling variable $r/\sqrt{t}$. The DNS data from two numerical simulations of the NS equation is shown as red and blue dots with error bars in Fig \ref{fig::NewDNSFit}, together with our theory (green curve) and the Kolmogorov $2/3$ prediction (black dashed line).

Our theory matches the data within a few $\%$ margin of error. The Kolmogorov scaling is totally off the charts.
The data index crosses the K41 value without any inertial range (the effective index plot is supposed to show a plateau around $2/3$ in the K41 model). Thus, there is no cascade in decaying turbulence (see \cite{migdal2024quantum} for a discussion of this issue).

Our solution has a well-defined continuum limit for observable variables such as decaying kinetic energy, energy spectrum, and moments of velocity difference, closely matching experiments.

At the same time, the dual system-- the string itself-- does not have any continuum limit. The regular star polygons with unit side have the radia $R = 1/(2 \sin( \pi p/q))$, which vary between $\oh$ and $\8$ but do not converge to any continuous function, with $p,q $ being co-prime numbers.
The distribution of the variable $X(p,q) =\cot(\pi p/q)^2/N^2 = \left( 4 R^2 -1\right)/N^2$ for large co-prime $1 \ge p<q < N$ was studied in the previous work, and it is a discontinuous piecewise power like distribution
 \begin{eqnarray}\label{CotDist}
   && f_X(X)= \left(1-\frac{\pi ^2}{675 \zeta (5)}\right)\delta(X) +\nonumber\\
   &&\frac{\pi^3}{3} X\sqrt{X}\Phi\left(\floor*{\frac{1}{\pi \sqrt{X}}}\right);
\end{eqnarray}
depicted in Figure.\ref{fig::PiPhi}.
Here $\Phi(n)$ is the totient summatory function
\begin{eqnarray}
    \label{PhiDef}
    \Phi(q) = \sum_{n=1}^q \varphi(n)
\end{eqnarray}
\pctPDF{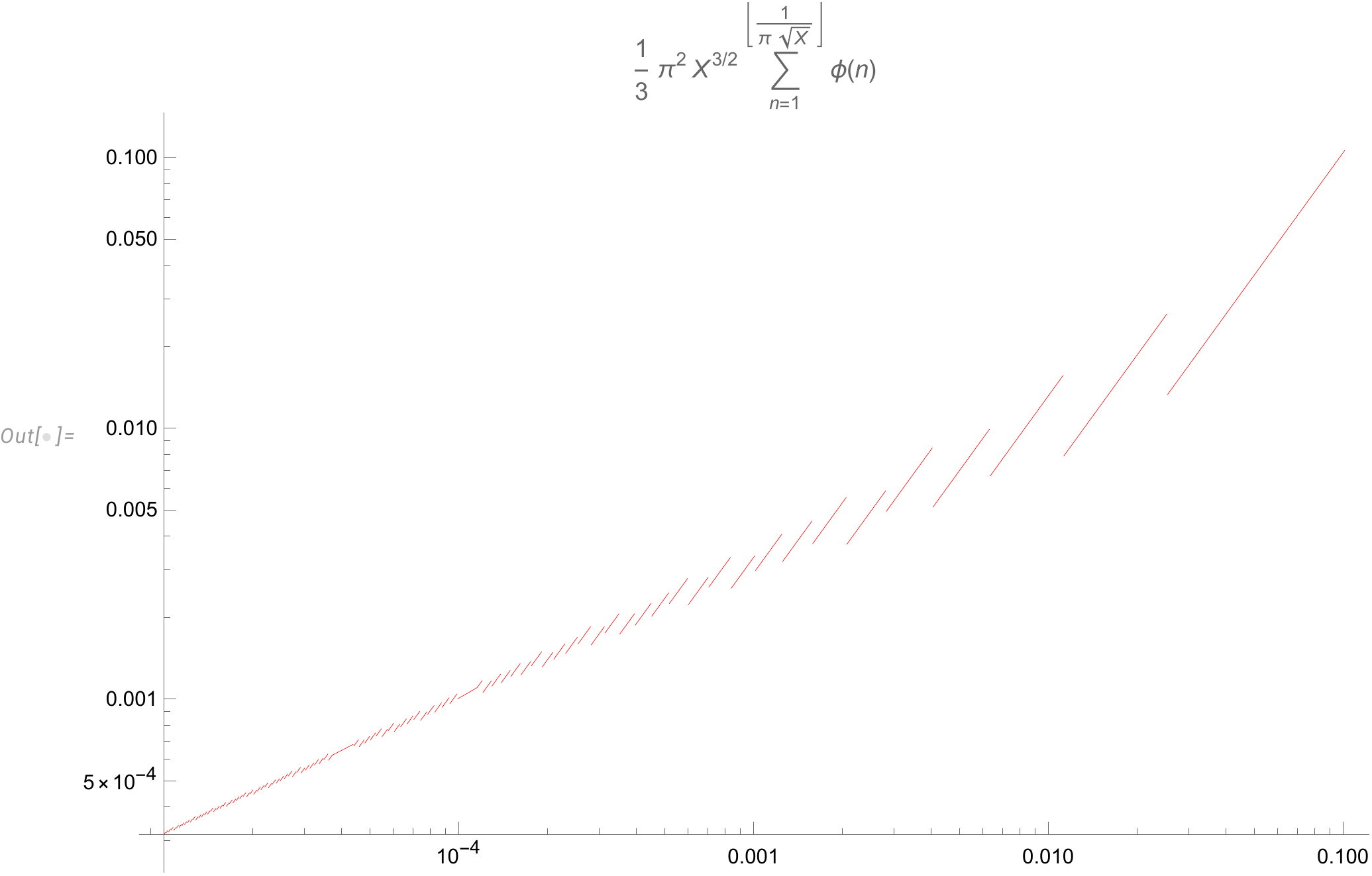}{Log log plot of the distribution \eqref{CotDist}}
This manifests the same phenomena we observed in the simplest solvable case of constant global rotation. This stationary solution of the NS equation studied in the section\ref{RandRot}  at finite $N$ is a finite sum of Gaussian random variables. However, there is no limit $N\to \8$ in this sum. Only the loop functional, related to the variance of the random variable $\vec P(\theta)$ tends to a finite limit related to the tensor area $\Sigma[C] = \oint_C \vec r \times d \vec r$.

In our solution of decaying turbulence, there are no random Gaussian variables: there's a random walk on random regular star polygons, equivalent to a string theory with discrete target space, which does not possess any continuum limit. At the same time, the dual amplitudes of this string theory have a continuum limit, providing the solution for the loop functional.

\subsection{An open problem of the stability of Euler ensemble as MLE fixed point}

The interesting and unexpected property of the Euler ensemble solution of the MLE is its independence of the spectral parameter $\gamma$. The $\gamma$ dependence reappears in the linearized MLE for the small deviations $\delta \vec F$ from the fixed point. These deviations describe the approach of the solution of the MLE to the fixed trajectory of decaying turbulence.

As we found in the first paper \cite{migdal2023exact}, these deviations decay by power laws with some indexes, depending on $\gamma$
\begin{eqnarray}
    \delta \vec F^{(i)}(\theta) \propto \psi_i(\theta|\gamma) t^{-\mu_i(\gamma)}
\end{eqnarray}
The spectral equation for these decay indexes $\mu_i(\gamma)$ was written down in \cite{migdal2023exact} for the finite $N$ in the Euler ensemble. The problem of the continuum limit of this spectrum is yet to be solved.

\section{Inconsistency of explosive solution}\label{MillProb}
Within our dual theory, there is, in principle, a possibility for finite-time explosion with $\vec F(\theta) \to \infty$ at some finite moment $\tau_c(\theta)$. 

In that case, only the third-degree terms will remain on the right side, with the linear term becoming negligible at $\tau \to \tau_c(\theta)-0$. The scale invariance fixes the time dependence in this case, so the solution becomes
\begin{eqnarray}
\label{Blow}
    &&\vec F(\theta, \tau) \to (\tau_c(\theta)- \tau)^{-\oh} \vec f(\theta);
\end{eqnarray}
We assume that the trajectory of singularity $\tau_c(\theta)$ is a continuous function of $\theta$ or a constant. In this case, all the terms on the right of the equation \eqref{FP} have a common time dependence $ (\tau_c(\theta)- \tau)^{-\frac{3}{2}}$, matching the left side.
The vector function $\vec f(\theta)$ must satisfy the following equation:
\begin{eqnarray}
    && \left((\Delta \vec{f})^2 +1\right)\vec f   =\Delta \vec{f} \left(\gamma^2 \vec{f} \cdot \Delta \vec{f} +\I \gamma \left( \frac{(\vec{f} \cdot \Delta \vec{f})^2}{\Delta \vec{f}^2}- \vec{f}^2\right)\right)
\end{eqnarray}
The left side of this equation for  $\vec f(\theta)$ differs from the left side of the equation \eqref{FP} for the fixed point for $\vec F$.

The following theorem proves the lack of a solution for this fixed point $\vec f$.
\begin{mytheorem}
\textbf{There is no explosive solution to the MLE with singularity position being a continuous function of the angle}.
\end{mytheorem}
\begin{proof}
Let us assume such a solution with some vector function $\vec f(\theta)$ and arrive at a contradiction.
    This vector equation is a linear combination of two vectors $a \vec f  = b \Delta \vec{f}$. Both coefficients $a,b$ must be zero. Otherwise, these two vectors are parallel, or else one of them vanishes. In both cases, the vorticity at the loop vanishes $\vec \omega(\vec C) \propto \vec f \times  \Delta \vec f =0$ at every point $\theta$ on the unit circle. Without vorticity, the solution reduces to the trivial fixed point $\Psi(\gamma, C) = 1$.

Now, the first coefficient $a$ can only vanish if  $\Delta \vec f$ has some imaginary component, which contradicts the requirement that the circulation $\oint d \vec C(\theta) \cdot \vec f(\theta) $ is a real variable. 

This requirement allows for a constant imaginary term in $\vec f(\theta) = \vec f_R(\theta) + \I \vec c$, as this continuous term will drop in the closed loop integral.
This requirement implies real discontinuity $\Delta \vec f$. In the explosion equation \eqref{Blow} with $a = (\Delta \vec{f}_R)^2 +1 >1$, there is no real solution with $a=0$.
\end{proof}

We have proven the inconsistency of the finite-time explosion in the momentum loop dynamics, i.e., the \NS{} dynamics with noisy initial data and constant or vanishing velocity at infinity. 

 This inconsistency is a consequence of the universality and dimensional reduction of the dual fluid dynamics, leading to much more stringent conditions on a potential explosion solution, which we have proven inconsistent.

In the conventional approach to the \NS{} equation, without the noise in initial data, Constantin and Fefferman have proven a theorem about the solution's regularity \cite{Constantin1993}. As a consequence of this theorem, any singular solution must have vorticity growing to infinity at some point in time in some region in space.

 In the MLE equation, vorticity at the loop would have a finite time  singularity with the above hypothetical solution 
\begin{eqnarray}
&&\VEV{\vec \omega(\vec C(\theta))\exp{\frac{\I \gamma \Gamma(C ,v)}{\nu}}} \nonumber\\
&&\propto \frac{1}{\tau_c(\theta)-\tau} \VEV{\vec f(\theta) \times \Delta \vec f(\theta) \exp{\I  \oint d \vec C(\theta')\frac{\vec f(\theta') }{\sqrt{2\nu(\tau_c(\theta')-\tau)}}}}_{f}
\end{eqnarray}
In particular, the mean square of vorticity (so-called enstrophy) would have a double pole
\begin{eqnarray}
    \VEV{\vec \omega(\vec C(\theta))^2}\propto \frac{\VEV{\left(\vec f(\theta) \times \Delta \vec f(\theta)\right)^2\exp{\I  \oint d \vec C(\theta')\frac{\vec f(\theta') }{\sqrt{2\nu(\tau_c(\theta')-\tau)}}}}}{(\tau_c(\theta)-\tau)^2} 
\end{eqnarray}
The growth of vorticity was proven necessary for the singular solution of \NS{} equation in \cite{Constantin1993}, and in our theory, it is ruled out.
   
If proven to stay in the smooth limit $\sigma \to 0$, without extra condition of continuous $\tau_c(\theta)$ this proof would provide a negative answer to the notorious problem of the explosion in the \NS{} equation, leaving two remaining alternatives: smooth (laminar) solution and a stochastic (turbulent) solution which we have found before \cite{migdal2023exact, migdal2024quantum} and reinterpreted in this work as a string theory.

Presumably, decaying turbulence occurs at a large enough Reynolds number in the initial data; otherwise, the solution stays smooth. 

\section{Discussion}\label{future}

The discovery of a new connection between different branches of science often lays the groundwork for unifying theories. Here, we outline potential generalizations of our findings on the equivalence between Navier-Stokes (NS) turbulence and random walks on discrete manifolds. These extensions open exciting directions for both physical applications and mathematical exploration.

\subsection{Physical Generalizations}

The applicability of the loop equations extends to other nonlinear systems that exhibit turbulence or finite-time singularities. Below, we propose several potential extensions, organized by increasing complexity:

\begin{itemize}
\item \textbf{Large-scale numerical simulations and quantum oscillations.} 
Our solution for the energy spectrum and velocity moments reveals an infinite 
set of decay indices, generalizing traditional multifractal scaling laws. 
Importantly, the presence of complex decay indices leads to oscillatory 
corrections, which represent a significant deviation from classical scaling 
predictions. These oscillations constitute a key and novel phenomenon: 
macroscopic quantum effects in classical fluids. Forthcoming large-scale 
numerical simulations will be crucial to clearly identify and verify these 
predicted oscillatory behaviors.

    \item \textbf{Turbulence forced by random rotations.} In this analytically solvable case, the loop equations, modified to include random centrifugal forces, exhibit a fixed point describing the steady state of forced turbulence with energy flow \cite{migdal2025inprep}.
    
    \item \textbf{Compressible fluids.} Adapting the framework to compressible flows requires replacing the incompressibility condition with variable-density dynamics, governed by the conservation of the volume element.
    
    \item \textbf{Magnetohydrodynamics (MHD).} In MHD, the circulation variables naturally split into two components encircling vorticity and magnetic fluxes. This generalization produces richer loop equations and leads to new insights into turbulence in plasmas. Preliminary analysis \cite{Migdal2025MHD} reduces the solution of MHD turbulence to the synchronized random walks on two regular star polygons.
    
    \item \textbf{Passive scalars.} The advection of passive scalar fields in a turbulent flow can be expressed via path integrals involving the loop functional. These integrals suggest analogies with gauge theories, offering new perspectives on scalar turbulence statistics.
    
    \item \textbf{General relativity.} Could turbulence mechanisms regularize naked singularities in Einstein's equations? While speculative, we propose studying the classical Einstein loop equations for stochastic effects. Analogous to the NS case, such stochastic mechanisms may provide an alternative to singularities, hinting at deeper connections between fluid dynamics and gravity.
\end{itemize}

\subsection{Mathematical Directions}

In addition to physical extensions, our framework suggests new mathematical avenues:

\begin{itemize}
    \item \textbf{Classification of dual PDEs.} Identifying classes of partial differential equations (PDEs) with dualities to quantum mechanical systems in loop space could deepen our understanding of nonlinear systems.
    
    \item \textbf{Dimensional reduction.} Exploring which dual PDEs can be reduced to one-dimensional nonlinear momentum loop equations may simplify otherwise intractable problems while preserving essential dynamics. The solution of the NS loop equations for arbitrary space $\mathbb R_d$ with dimension $d \ge 3$ was already found in \cite{migdal2023exact}. This solution does not exist for $d < 3$, but for larger dimensions is essentially the same: the random walk on regular star polygons, with random rotation of the plane in $\mathbb R_d$.
    
    \item \textbf{Nonlinear spaces.} Generalizing the observed random walks on star polygons to loop groups and other nonlinear spaces may reveal new connections between geometry, number theory, and fluid mechanics.
    \item \textbf{Summing over topologies?} The turbulent flows can exist (at least mathematically) in spaces of arbitrary topology, for example, on a manifold of higher genus. The natural generalization of the discrete string theory made of regular star polygons also allows topological transition when a $(q,p)$ polygon with a winding number $n$ bifurcates into two $(q,p)$ polygons with winding numbers $k, n-k$, which later evolve the same way as before. The opposite process of merging $k, m \Rightarrow k+m $ is also possible without breaking any features of preceding or subsequent propagation. How to sum our WKB solution of the Euler ensemble over topologies? How are they related to decaying turbulence in "physical" space?
\end{itemize}

\section{Conclusion}
In this paper, we introduced a novel theoretical framework connecting fluid dynamics 
to a solvable nonlinear dynamical system in loop space. The main contributions include:

\begin{itemize}
    \item A rigorous reformulation of the three-dimensional Navier–Stokes equations 
    into a solvable one-dimensional nonlinear loop-space equation.

    \item The \textit{No Explosion Theorem}, rigorously excluding finite-time singularities 
    for Navier–Stokes flows initiated with stochastic velocity field due to thermal fluctuations.

    \item The introduction and validation of an exact analytical solution, termed the 
    \textit{Euler ensemble}, describing the universal asymptotic state of decaying turbulence. 
    This solution is strongly supported by numerical simulations and experimental data.
    
    \item A demonstration of explicit mathematical equivalence between the Euler ensemble 
    solution and a solvable discrete string theory formulation.
\end{itemize}

These results indicate that turbulent flows, despite their complexity and apparent randomness, 
can possess universal structures that are analytically describable and rigorously verifiable. 
Our findings provide a solid foundation for future theoretical, numerical, and experimental 
studies of turbulence.


\appendix
\setcounter{equation}{0}
\section*{Acknowledgments}
I benefitted from discussions of this work with Camillo de Lellis, Elia Bruè, Stan Palasek, Semon Rezchikov and Jincheng Yang. The most valuable advice came from Albert Schwartz.
This research was supported by the Simons Foundation award ID SFI-MPS-T-MPS-00010544 in the Institute for Advanced Study.

\section*{ORCID}
\noindent Alexander Migdal  - \url{https://orcid.org/0000-0003-2987-0897}
\appendix
\section*{Momentum Loop distribution for noisy velocity}\label{PolygonalW}

In the polygonal approximation, the path integral \eqref{InvFourier} becomes a multiple Fourier integral
\begin{eqnarray}
    &&W(P_1,\dots  P_N) = \int \prod_k d^3 \vec C_k \delta^3(\vec C_N- \vec C_1) \exp{- m_0 |C| +\I \sum_k \vec{C}_k \cdot \Delta\vec{P}_k};
\end{eqnarray}
We rewrite it as a multiple integral over steps $ \vec\eta_k =\Delta \vec C_k $
\begin{eqnarray}
    &&W(P_1,\dots  P_N) = \int \prod_k d^3 \vec \eta_k \delta^3(\sum \vec \eta_k)  \exp{- \sum_k \left(m_0 | \vec \eta_k| - \I \vec{\eta}_k \cdot \vec{P}_k \right)};
\end{eqnarray}
Introducing a Fourier integral for the delta function we find
\begin{eqnarray}
    &&W(P_1,\dots  P_N) = \int d^3 q \int \prod_k d^3 \vec \eta_k  \exp{- \sum_k \left(m_0 | \vec \eta_k| - \I \vec{\eta}_k \cdot (\vec{P}_k- \vec q)\right)};
\end{eqnarray}
Now the integral over $\eta_k$ is calculable ( I skip positive normalization factors, and use known angular integral $\oh\int_{-1}^1 d z \exp{ \I x z} = \frac{\sin x }{x} = \Im \exp{\I x}/x $):
\begin{eqnarray}
    &&\int d^3 \vec \eta_k  \exp{- \left(m_0 | \vec \eta_k| - \I \vec{\eta}_k \cdot (\vec{P}_k- \vec q)\right)} \propto
    \Im \int_0^\infty \eta d \eta \exp{- \left(m_0 \eta - \I \eta |\vec{P}_k- \vec q|\right)}/|\vec{P}_k- \vec q| \propto\nonumber\\
    &&\frac{m_0}{(m_0^2 + |\vec{P}_k- \vec q|^2)^2}
\end{eqnarray}
Collecting the factors we get a POSITIVE measure
\begin{eqnarray}
    W(P_1,\dots  P_N) \propto \int d^3 q \prod_k \frac{m_0}{(m_0^2 + |\vec{P}_k- \vec q|^2)^2}
\end{eqnarray}
In the limit of large $N$ this becomes a Gaussian measure
\begin{eqnarray}
    &&\int d^3 q \prod_k \frac{1}{(1 + |\vec{P}_k- \vec q|^2/m_0^2)^2} \to \nonumber\\
    &&\int d^3 q \exp{-2 \sum_k |\vec{P}_k- \vec q|^2/m_0^2} \propto \exp{- 2\sum_k |\vec{P}_k- \vec P_s|^2/m_0^2};\\
    &&\vec P_s = \sum \vec P_k/N;
\end{eqnarray}
This ensemble assumed fixed $N\to \infty$.  One can make this $N$ variable and introduce a weight  $\exp{-\mu N}, \mu \to 0$ (so called canonical ensemble as opposed to a micro-canonical with fixed $N$).
In that case, in the continuum limit, summation over $N$ becomes an integral $\int d T$ and we get
 \begin{eqnarray}
     W[P] \propto \int_0^\8 d T \exp{ - m^2 T -\int_0^T d t (\vec P(t) - \vec P_s)^2}
 \end{eqnarray}
 with corresponding factors absorbed into $T \propto N/m_0^2, m^2 \propto \mu m_0^2$. This is the distribution quoted in the text of the paper.








    
\bibliographystyle{ws-ijmpa}
\bibliography{bibliography}
\end{document}

%% file: mymacros.tex
\DeclareMathOperator{\sign}{sign}
\newcommand{\Mathematica}{\textit{Mathematica\textsuperscript{\resizebox{!}{0.8ex}{\textregistered}}}}
\def\8{\infty}
\def\oh{\sfrac{1}{2}}

\def\tt{\sfrac{2}{3}}

\newcommand*{\I}{\imath}%

\def\const{\textit{ const }}
\def\undertext#1{\vtop{\hbox{#1}\kern 1pt \hrule}}

\def\abs#1{\left| #1\right|}

\def\pd#1{\partial_{#1}}
\def\VEV#1{\left\langle #1\right\rangle}

\def\ff#1{\frac{\delta}{\delta#1}}
\def\fbyf#1#2{\frac{\delta#1}{\delta#2}}

\def\bea{\begin{eqnarray} && &&}
\def\eea{\end{eqnarray}}

\let\oldexp\exp
\renewcommand{\exp}[1]{\oldexp\left(#1\right)}

\def\RE{\textbf{Rey}}
\def\NS{Navier-Stokes}

\def\BS{Biot-Savart}

\def\val{v_{\alpha}}
\def\vbe{v_{\beta}}
\def\vga{v_{\gamma}}

\def\rbe{r_{\beta}}

\def\oal{\omega_{\alpha}}
\def\obe{\omega_{\beta}}
\def\oga{\omega_{\gamma}}

\newcommand{\Mod}[1]{\ (\mathrm{mod}\ #1)}

\def\XXint#1#2#3{{\setbox0=\hbox{$#1{#2#3}{\int}$}
     \vcenter{\hbox{$#2#3$}}\kern-.5\wd0}}

\newcommand{\bS}{\mathbb{S}}

\renewcommand{\Re}{\textbf{Re }}
\renewcommand{\Im}{\textbf{Im }}

\newcommand{\tpmod}[1]{{\@displayfalse\Mod{#1}}}
\usepackage[english]{babel}

\usepackage{mathtools}

\DeclarePairedDelimiter\floor{\lfloor}{\rfloor}

\newcommand{\pctPDFRot}[2]{
\begin{figure}
    \centering
    \includegraphics[angle=-90,width=1.0\textwidth]{#1.pdf}
    \caption{#2}
    \label{fig::#1}
\end{figure}
}

\newcommand{\pctPDF}[2]{
\begin{figure}
    \centering
    \includegraphics[width=\textwidth]{#1.pdf}
    \caption{#2}
    \label{fig::#1}
\end{figure}
}
\newcommand{\pctWPDF}[3]{
\begin{figure}
    \centering
    \includegraphics[width=#1\textwidth]{#2.pdf}
    \caption{#3}
    \label{fig::#2}
\end{figure}
}

\newtheorem{mytheorem}{Theorem}